\PassOptionsToPackage{dvipsnames}{xcolor}
\documentclass[sn-mathphys]{sn-jnl}

\usepackage{stmaryrd}

\jyear{2021}%

\theoremstyle{thmstyleone}%
\newtheorem{theorem}{Theorem}
\theoremstyle{thmstyletwo}%
\newtheorem{example}{Example}%

\theoremstyle{thmstylethree}%
\newtheorem{definition}{Definition}%

\definecolor{subspace}{RGB}{125, 5, 10}
\definecolor{freeze}{RGB}{25, 30, 100}
\definecolor{fsr}{RGB}{120, 3, 50}
\definecolor{reset}{RGB}{80, 10, 170}

\begin{document}

\title[The Ising model and random fields of scales]{The Ising model and random fields of scales}


\author{\fnm{Ricardo} \sur{G\'omez A\'iza}}\email{rgomez@im.unam.mx}

\affil{\orgdiv{Instituto de Matem\'aticas}, \orgname{Universidad Nacional Aut\'onoma de M\'exico}, \orgaddress{\street{Circuito Exterior, Ciudad Universitaria}, \city{Ciudad de M\'exico}, \postcode{04510}, \state{CDMX}, \country{M\'exico}}}


\abstract{Random fields of scales result when the class of musical scales is thought as a set of sites,
and a site can be in one of two possible states or ``spins'': \textsf{On} or \textsf{Off}. We
present a flexible simulated annealing model that produces generic
configurations arising from equilibrium states (or Gibbs measures)
associated to hamiltonian energy functions defined in terms of musical
interactions with parameters that can be manipulated
to customize properties of the scales. The starting point is to think of the set of scales
as the combinatorial class of integer compositions and the final result is an effective
thermodynamic search engine implemented in an open access application
for the 12-TET tuning system: \emph{Scaletor}.}

\keywords{musical scales; integer compositions; Boltzmann machines; Montecarlo Markov chain; Glauber dynamics}


\pacs[MSC Classification]{37B10; 05A15; 00A65; }

\maketitle


\section{Introduction}\label{sec:intro}

We present a mathematical model implemented in a software application
designed to classify subsets of musical scales within the $n$-TET tuning system.
In this framework, a scale is viewed as a site, and each site can be in either one of
two possible states or ``spins'': \textsf{On} or \textsf{Off}. The model is inspired by
thermodynamic formalism, particularly the Ising model, one of the most well-known
statistical mechanics models for interacting particle systems.
We adapt its foundational principles for the scenario is that of a finite set of sites 
connected by several types of networks that we use to define interactions 
in terms of ``musical energies''. These interactions give rise to
Hamiltonian functions, from which equilibrium states
--formally described as Gibbs measures--
emerge and are realized as Boltzmann distributions.

Under this premise, given a set of parameters provided by the user,
we can simulate the corresponding generic configurations
with algorithms like Glauber dynamics and Metropolis-Hastings for the Ising model,
and then the temperatures can be
manipulated, as in simulated annealing processes,
to cross phase transitions and observe emerging order realized as
configurations of scales produced by measures of maximal pressure in the
different phases of the given parameterizations.
The Gibbs measures can be Dirac probability measures
supported on single configurations of scales with specific properties and
our simulations can produce them accurately if properly calibrated.
The properties can be combinatorial, e.g. can deal with the type
and number of intervals that form the pitch classes of the scales,
or they can deal with the modes of the scales, or they may take
into account interpolation,
or they can be of geometric type, e.g. can deal with the location of
the center of balance of the scales (see \cite{Carey17,MilneBulgerHerff17}), etc.
The musical interactions and the simulated annealing processes
were designed to be controlled in an ``audio mixer'' type of component
in order to gain intuition when doing manipulations.

The application, henceforth called \emph{Scaletor},
was initially experimentally developed in \emph{Processing} (see \cite{WebP}),
a platform with the capability of producing standalone
java applications for several operating systems.
So as a complementary material accompanying this work,
\medskip

\centerline{\emph{Scaletor is an open access application.}\footnote{See section \ref{sec:additional} for more information about additional material and download instructions (in particular, systems capable of running Processing's standalone applications are required).}}
\medskip

\noindent
We hope that it will
serve both the interested readers to verify the claims in this work
and the musicians and music theorists for it is in fact a complete
catalog of all the 2048 musical scales
(each scale has been named, mainly following the 
nomenclature in \cite{UESMDecks20}). The 
thermodynamic search engine in \emph{Scaletor}
is elaborated but ultimately effective and the
eventual experienced user can find it handy.
A screenshot of the Graphic User Interface (GUI) 
is shown in Figure \ref{fig:scaletor} (this figure contains a rather large caption with concepts
that later in the text will be described in detail).

\begin{figure}[h] 
   \centering
   \includegraphics[width=4.65in]{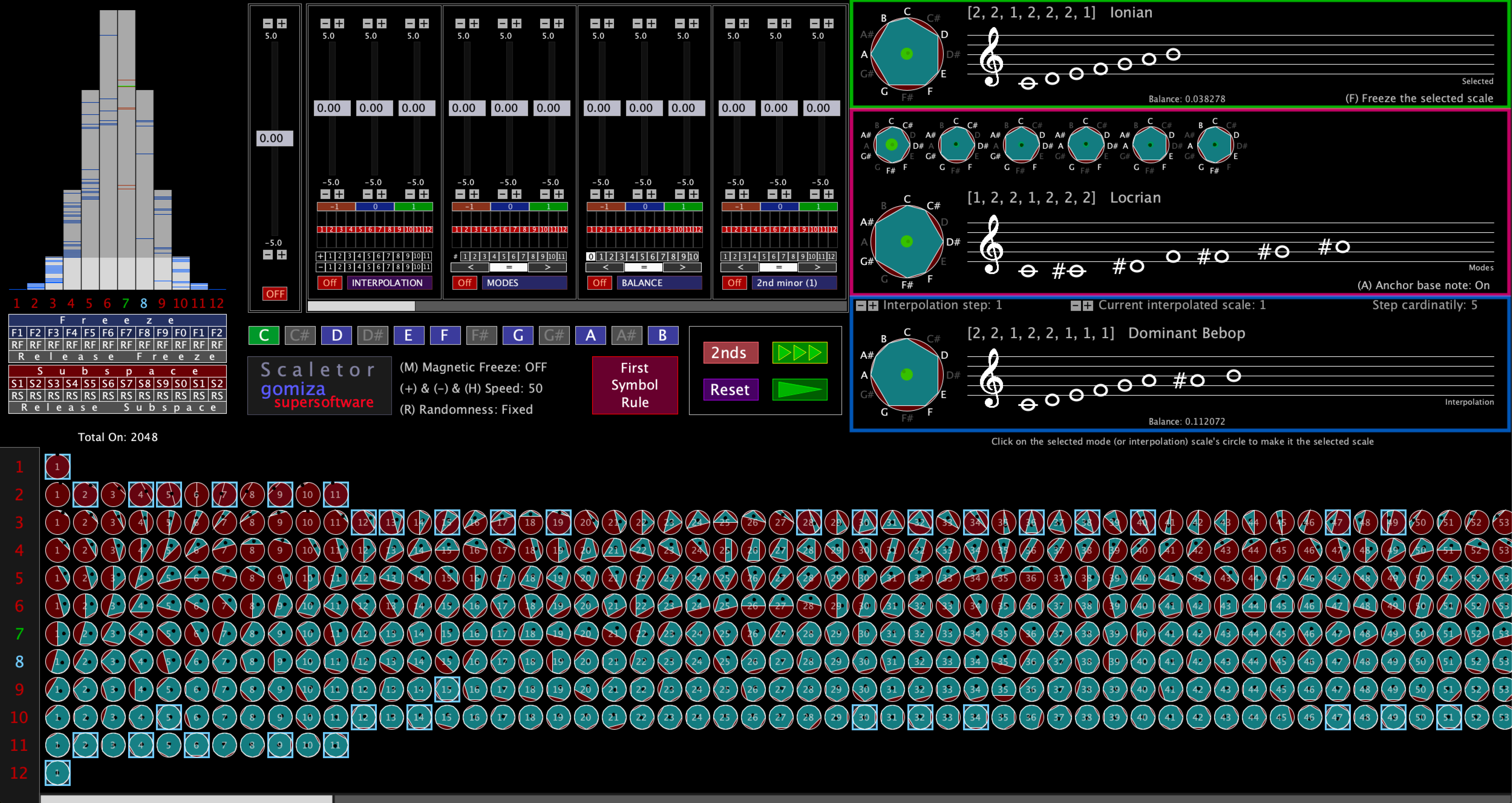}
   \caption{A screenshot of the GUI of \emph{Scaletor} on startup.
   The bell of scales is shown at the upper left corner.
   Each horizontal bar in it
   represents a scale in the selected tonality (in this case it is 
   \textsf{C}, in the green box near the bottom-right
   part of the bell of scales).
   The colors in the
   bell of scales code information about the status of several things in \emph{Scaletor} at
   each instant of time, and some of this information
   is shown along other parts of the GUI. For example,
   all the scales are in the \textsf{On} state because there are no black bars in the bell of scales.
   Also, the selected scale is the Ionian scale $s=(2,2,1,2,2,2,1)$ and it corresponds to the
   green bar in the 7th column in the bell of scales (zoom in), it is also shown in detail on the upper-right
   part of the GUI. The selected Ionian scale has 6 other modes, they correspond to the pink bars
   in the bell of scales, all of them are also in the 7th column of course,
   the first of them is shown in detail on the right part of the GUI, under the selected scale, it is
   the Locrian mode $(1,2,2,1,2,2,2)$. The Dominant Bebop scale $(2,2,1,2,2,1,1,1)$
   shown in detail under the Locrian mode  is the first blue bar in the 8th column
   of the bell of scales, there are 5 blue bars in this column and they correspond
   to the neighborhood of the selected scale in the $1^+$-step interpolation
   network $\mathfrak I^{1+}$, i.e there are 5 scales in \textsf{C} of length 8
   that result from the Ionian scale by adding a new pitch class.
   All the blue bars in the bell of scales correspond
   to the union of all the neighborhoods of the selected scale in the
   $k^+$-step and $k^-$-step interpolation networks $\mathfrak I^{k+}$
   and $\mathfrak I^{k-}$ resp., for all $k\geq 1$, i.e.
   all the blue bars correspond to the scales that result from the Ionian scale
   by either adding new pitch classes,
   or removing pitch classes other than the tonality \textsf{C}.
   The bottom-half part of the GUI represents a horizontal zoom of the
   bright part at the bottom of the bell of scales: here only the scales in the \textsf{On} state
   are shown in more detail as polygons, also with color codes as above
   (observe that at the very bottom there is a scroll-bar to navigate).
   The parameters that control the thermodynamic search engine are contained
   in the component that looks like and audio mixer.
   The application is silent: no scale is being played.
   There are no restricted nor frozen sites
   and the configuration is not changing because the MCMC random process is not running;
   the MCMC has the master volume at zero, and in fact it is muted (Off),
   hence the random field of scales that would occur if the MCMC were running
   would consist of independent \textsf{Bernoulli(1/2)} random variables (pure random noise). Etc.
   }
   \label{fig:scaletor}
\end{figure}

The rest of the paper is organized as follows.
In section \ref{sec:background} we give a very brief and basic background
that motivates our setup: we quickly recall the Ising model, first by
addressing spaces of configurations,
hamiltonian energy functions and Gibbs measures,
interactions, and finish with the Glauber dynamics and Metropolis-Hastings
simulations of generic configurations.
All this section is mainly intended to motivate the model
in \emph{Scaletor}, but it also serves to present the basics we need
to readers that may not be familiar with this type of material
(for deep and comprehensive references see
\cite{Georgii88, Ruelle04}). So, a reader familiar with thermodynamic formalism
can easily go directly to the main section \ref{sec:scaletor} where we present in detail
how we adapt the ideas behind the Ising 
model to the scenario of random fields of scales. In
section~\ref{sec:examples} we show examples of outputs
produced by the application to exhibit its performance.
Section \ref{sec:conclusion} contains
conclusions and related works (this type of models could be applied to
artificial intelligence systems specialized in music).
The last section \ref{sec:additional} contains information about 
\emph{Scaletor}'s repository at GitHub and a tutorial video.
\bigskip

\section{Background}
\label{sec:background}

\subsection{Boltzmann machines, hamiltonian energies and equilibrium states}

The type of model we develop here carries similarities with what are known
as Boltzmann machines, the Sherrington-Kirkpatrick model, Hopfield networks,
Potts model, Edwards-Anderson's model, the Ising model itself, etc.
\cite{Sherrington75, Edwards75, Chatterjee23}.
There is a finite set of \emph{sites} $\mathcal S$ (in our case it will be the set of scales in a given tonality, in the 12-TET tuning system)
and hence there is a (finite) \emph{spin configuration space}\footnote{The terminology comes from atomic spin models: in ferromagnetic materials, neighboring atoms tend to have aligned spins, either \emph{positive} or \emph{negative}, represented by the symbols $ \textsf{+1}$ and $ \textsf{-1}$, respectively
(if instead the material is \emph{anti}ferromagnetic, then the spins tend to be opposite). In our context, the symbols $ \textsf{+1}$ and $ \textsf{-1}$ will
represent the states \textsf{On} and \textsf{Off} that a given scale can be at a certain instant of time, respectively.} 
\begin{align}
	\Omega \subseteq \{\textsf{+1}, \textsf{-1}\}^{\mathcal S} \triangleq
	\left\{ \omega \colon \mathcal S \to \{ \textsf{+1}, \textsf{-1} \}\right\}
\end{align}
(as usual, for every $x\in \mathcal S$, we let $\omega_x \triangleq \omega(x)$).
Also, for each configuration $\omega \in \Omega$ there is the idea of cost, or weight,
captured by a hamiltonian function $E\colon \Omega \to \mathbb R$ so that
$E(\omega)$ is thought as the \emph{energy} required to realize $\omega$.
We look for typical configurations coming out of measures of maximal pressure,
that is, the generic configurations for the \emph{equilibrium states}
(or \emph{Gibbs measures}\footnote{For the equivalence between equilibrium states and Gibbs measures, see e.g. \cite{Georgii88, Ruelle04, BGMT20,BGMMT21}. In our context these two concepts are equivalent and we will content ourselves with the definition of equilibrium state and call it Gibbs measure indistinctively.}):
the probability distributions that solve the optimization problem
\begin{align}
\label{eq:max}
	\max \{\Psi(p) \: \rvert \: p \colon \Omega \to [0,1] \hbox{ is a probability function}\}\smallskip
\end{align}
where
\begin{align}
	H(p) \triangleq - \sum_{\omega \in \Omega} p(\omega)\log p(\omega)
\end{align}
is the \emph{entropy} of the probability distribution $p$
and
\begin{align}
	\Psi (p) \triangleq H(p) - \beta \int_{\Omega} E dp
\end{align}
is the corresponding \emph{pressure} at
the (\emph{inverse})
\emph{temperature}\footnote{In the thermodynamic theory of lattice gasses,
$\beta = 1 /k_BT \in \mathbb R$ where $T$ is the temperature
and $k_B = 1.380649 \times 10^{-23}$
is the Boltzmann constant, and here, as it is customary,
all these constants will be absorbed by the parameter $\beta$.} $\beta$.
The finitary variational principle asserts that there is
a unique equilibrium state, that is, there is a unique
probability measure that solves the optimization problem
\eqref{eq:max}, it is precisely
given by the Boltzmann distribution $p \triangleq \mu$ defined by
\begin{equation}
\label{eq:boltzman}
	\mu (\omega) = \displaystyle\frac{e^{-\beta E(\omega)}}{Z}
	\quad \forall \omega \in \Omega
\end{equation}
where $Z \triangleq \sum\nolimits_{\omega \in \Omega} e^{-\beta E(\omega)}$ is the
normalizing constant known as the \emph{partition function} (the proof follows from
Jensen's inequality).

\subsection{Interactions}

The hamiltonian energy function $E$ returns
the total cost of energy required to realize a given configuration $\omega \in \Omega$,
and such amount of energy is thought as the result of adding
all the energies arising from the interactions that occur between
the constituents of $\omega$. Formally,
an \emph{interaction} is a family
$\Phi=\{\Phi_A\}_{A\Subset \mathcal S}$
of \emph{local energy functions}
$\Phi_A\colon\Omega\to \mathbb R$ indexed by all the finite\footnote{In this work, $\mathcal S$ is always finite, hence any subset of it is also finite, nevertheless we stress the fact that $A$ is required to be finite (this is the meaning of the symbol $\Subset$) because the main theory is for contexts when $\mathcal S$ is infinite (when the so called ``thermodynamic limit'' comes into play).} subsets, or \emph{regions}, $A\Subset \mathcal S$.
This means that $\Phi_A$ solely depends on the configurations restricted to $A$, or to be precise,
$\Phi_A(\omega) = \Phi_A(\omega')$ whenever $\omega,\omega'\in\Omega$ are such that
$\omega \rvert_A=\omega' \rvert_A$.
The interaction $\Phi$ induces a hamiltonian energy function
$E=E_{\Phi}$ which is the sum of all the local
interaction energies, namely
\begin{equation}
\label{eq:interhamilton}
	E_{\Phi}(\omega)\triangleq \sum\limits_{A\Subset \mathcal S}\Phi_A(\omega)
	\quad \forall \omega \in \Omega.
\end{equation}
To construct custom interactions it is appropriate to
start with interactions that act on simple regions of the state space,
and then consider operations like addition, multiplication and
composition with real functions.
Let $\Pi(\Omega)$ be the set of interactions on $\Omega$.
For every
$\Phi = \{ \Phi_A \colon \Omega \to \mathbb R\}_{A \Subset \mathcal S} \in \Pi(\Omega)$ and
$\Psi = \{ \Psi_A \colon \Omega \to \mathbb R\}_{A \Subset \mathcal S} \in \Pi(\Omega)$,
let
$\Phi + \Psi \in \Pi(\Omega)$ and
$\Phi \cdot \Psi \in \Pi(\Omega)$ 
be defined by
$(\Phi+\Psi)_A(\omega) = \Phi_A(\omega)+\Psi_A(\omega)$
and
$(\Phi\cdot\Psi)_A(\omega) = \Phi_A(\omega)\cdot\Psi_A(\omega)$
for all $A\Subset \mathcal S$ and $\omega\in\Omega$.
Also, given any function $f\colon \mathbb R \to \mathbb R$
(e.g. $f(x)=ax+b$ with $a,b \in \mathbb R$),
let $f(\Phi)\in\Pi(\Omega)$ be defined by
$f(\Phi)_A(\omega) = f\big(\Phi_A(\omega)\big)$ for all
$A\Subset \mathcal S$ and $\omega\in\Omega$.
The addition of two interactions $\Phi$ and $\Psi$
is \emph{disjoint} if $\Phi_A=0$ whenever $\Psi_A\neq 0$
and viceversa. The interaction $\Phi$ is
\emph{single site} if $\Phi_A=0$ whenever $A$ is not a singleton,
i.e. $\Phi_A=0$ if $A\neq \{x\}$ for some $x\in\mathcal S$.
Single site interactions produce independent random fields.

Now the goal is to simulate generic configurations in $\Omega$ for the
Gibbs measures $\mu = \mu_{\Phi}$ associated to hamiltonian energy functions $E_\Phi$ arising
from interactions $\Phi \in\Pi(\Omega)$. This is possible by means of algorithms like Glauber dynamics or Metropolis-Hastings for the Ising model. So let us first briefly present the Ising model
as example.

\subsection{The Ising model}

In the classical Ising model of size $n\geq 1$
in dimension $d \geq 1$, the set of sites is the $d$-dimensional square grid of integers modulo $n$,
that is,
\begin{align}
	\mathcal S = \mathbb Z_n^d \triangleq \underset{d\textnormal{ times}}{\underbrace{\mathbb Z_n \times \cdots \times \mathbb Z_n}}
\end{align}
where $\mathbb Z_n = \mathbb Z / n\mathbb Z$ is the additive group of integers modulo $n$.
Also the spin configuration space is the complete set of configurations
\begin{align}
\label{eq:initialIsing}
	\Omega = \{\textsf{+1},\textsf{-1}\}^{\mathbb Z_n^d}.
\end{align}
The symbols $\textsf{+1}$ and $\textsf{-1}$
that represent the states in which a site can be are thought as the actual integer
values $+1$ and $-1$, respectively.
There is the notion of adjacency in $\mathbb Z^d_n$
defined by the rule that makes two elements $x,y\in \mathbb Z^d_n$ adjacent if and only
if there are representatives of $x$ and $y$ in $\mathbb Z^d$ that are at euclidian distance one.
In other words, the \emph{neighborhood} of a site $x \in \mathbb Z_n^d$ is
$N(x) \triangleq \{ x + \mathbf e  : \mathbf e \in N(\mathbf 0) \}$,
where
$N(\mathbf 0) \triangleq \{ \mathbf e \in \{-1,0,1\}^d : \| \mathbf e \| = 1 \}$
($\|\cdot\|$ denotes euclidian distance).
The well known hamiltonian
energy function for the two-dimensional Ising model ($d=2$) takes into account this adjacency
and is defined for every $\omega \in \Omega$ by
\begin{align}
\label{eq:IsingHamiltonian1}
	E(\omega) = & -\sum\limits_{x \in \mathcal S} \sum\limits_{y \in N(x)}
	J(x,y)\omega_x \omega_y \\ \label{eq:IsingHamiltonian2}
	& - h \sum\limits_{x \in \mathcal S} \omega_x 
\end{align}
where $J(x,y)=J(y,x) =J \in \mathbb R$ is a (uniform) symmetric assignment
of ``weights'' (real numbers) on the edges of the $\mathbb Z_n^d$
grid\footnote{This assignment of weights
can be coded in a square $n^d \times n^d$ symmetric matrix that has nonzero
entries only where the adjacency matrix of the $\mathbb Z_n^d$ grid has nonzero entries.},
they represent the interaction energy of the states $\omega_x$ and $\omega_y$ of
neighboring particles at sites $x,y\in\mathcal S$ (so that
if $J > 0$, then the neighboring particles tend to have their spins aligned and the system
is \emph{ferromagnetic}, otherwise, if $J < 0$, then neighboring particles tend to have
opposite spins and the system is \emph{antiferromagnetic}), and $h \in \mathbb R$ is to be interpreted as an \emph{external magnetic field}\footnote{In the context of lattice gases, the particles attract or repeal each other, and
the symbols $\textsf{+1}$ and $\textsf{-1}$ are thought as the integer values $+1$ and $0$,
respectively, and represent a particle that occupies or not a site, respectively, and $h$ is thought as a \emph{chemical potential}.}.

One way to define an interaction $\Phi\in\Pi(\Omega)$ that gives rise to $E$ as given in equations
\eqref{eq:IsingHamiltonian1} and \eqref{eq:IsingHamiltonian2} is by
first to declare that $\Phi_A = 0$ whenever
$A$ is not a single site nor a site together with its neighborhood, and then let
\begin{align}
\label{eq:IsingInter1}	
	\Phi_{\{x\}\cup N(x)}(\omega) 
	& = -J\omega_x\sigma_x(\omega) \\
\label{eq:IsingInter2}
	\Phi_{\{x\}}(\omega) & = - h \omega_x 
\end{align}
for every $x \in \mathcal S$,
where $\sigma_x(\omega) \triangleq \sum\limits_{y\in N(x)}\omega_y$. Thus, from equation \eqref{eq:interhamilton},
we readily get $E_{\Phi} = E$.

Above, the family of functions
$\Phi^{(\textrm{I})}=\{ \Phi_{\{x\}} \colon \Omega \to \mathbb R \}_{x\in \mathcal S}$
defined by equation \eqref{eq:IsingInter2} is an instance of a single site
interaction\footnote{When we say that a family $\Phi$ of functions
$\Phi_A\colon\Omega \to \mathbb R$
is an interaction but such a family is incomplete in the sense that there are regions
$A\Subset \mathcal S$ such that $\Phi$ possesses no functions $\Phi_A$ that have been explicitly defined, we
are thinking in the completion obtained by adding the functions $\Phi_A=0$ in such cases.},
it is quite simple what each of these functions do, e.g. if $h=1$, then
$\Phi_{\{x\}} (\omega)$ is the opposite state of $\omega_x$.
On the other hand, consider the family
$\Phi^{(\textrm{II})}=\{ \Phi_{\{ x\} \cup N(x)}:\Omega \to \mathbb R \}_{x\in \mathcal S}$
of functions defined by \eqref{eq:IsingInter1}, it is not a single site interaction
and so the states in configurations are generally dependent of each other
(unless $J=0$ of course). Observe that $\Phi$ is the (disjoint) sum of both $\Phi^{(\text{I})}$
and $\Phi^{(\text{II})}$. Figure \ref{fig:isingInteraction}
shows all the posible values of $\Phi_A$ for
all the possible configurations on regions $A= \{x\}\cup N(x)$ with $x\in\mathcal S$
(with $J=1$).

\begin{figure}[h] 
   \centering
   \includegraphics[width=4.7in]{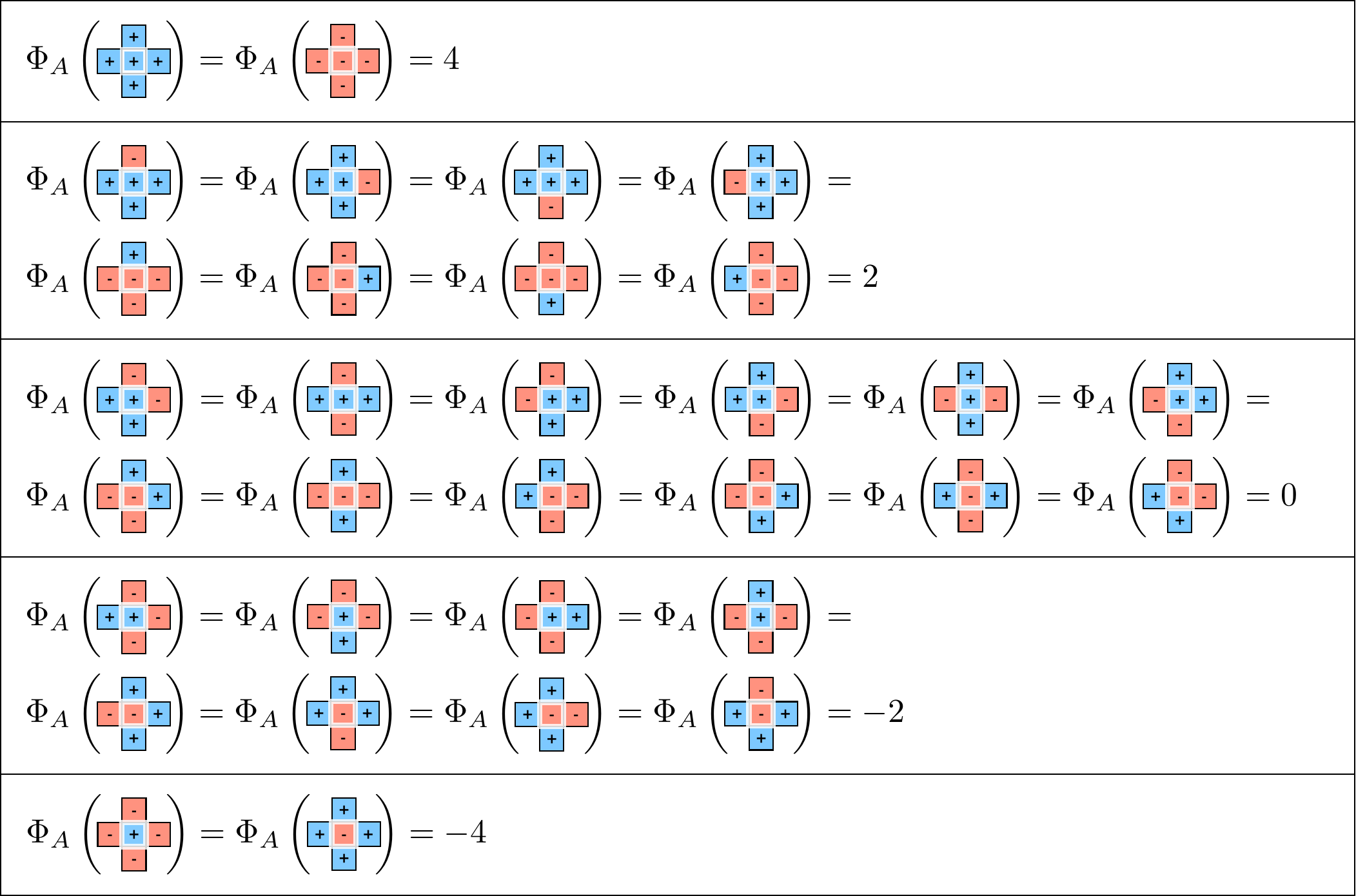} 
   \caption{The values of $\Phi_A$ in the Ising model
   over all configurations on regions $A= \{x\}\cup N(x)$
   formed by a site $x\in\mathcal S$ and its neighborhood
   (here we have taken $J=1$, see equation \eqref{eq:IsingInter1}).}
   \label{fig:isingInteraction}
\end{figure}

\subsection{Glauber dynamics and Metropolis-Hastings simulations}
\label{subsec:GlauberMCMC}
Glauber dynamics is a Monte Carlo Markov Chain (MCMC)
simulation that produces generic
configurations of the Gibbs measures in the classical Ising model on
$\mathcal S = \mathbb Z_n^2$, with the interaction $\Phi$ consisting
of local energy functions as defined above in equations
\eqref{eq:IsingInter1} and \eqref{eq:IsingInter2}.
It starts from an arbitrary random configuration
$\omega \in \Omega$ (e.g. generated from a
random field indexed by $\mathcal S$ and
formed by independent \textsf{Bernoulli}$(1/2)$ distributions
on each site $x \in \mathcal S$). Then the simulation process is as follows
(see Figure \ref{fig:isingGlauber}):
\begin{figure}[h] 
   \centering
   \includegraphics[width=4.7in]{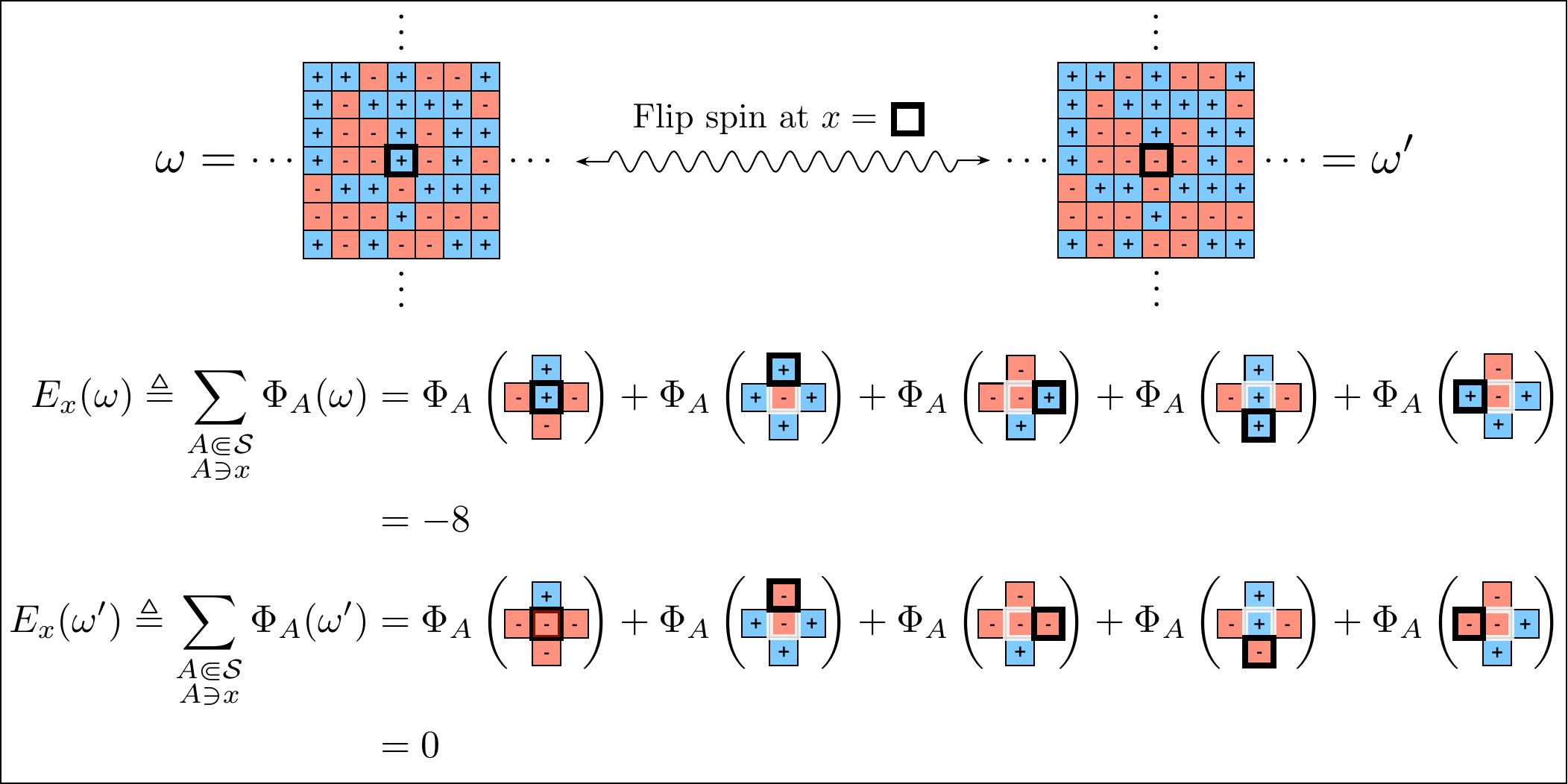} 
   \caption{An instance of two configurations that differ only on one of its sites
   and their corresponding local energies. Here, if the spin at $x$
   were to flip from \textsf{+1} to \textsf{-1}, then the change in energy from
   $\omega$ to $\omega'$ would be
   $\Delta E_x(\omega)=8$
   (see equations \eqref{eq:local0} and \eqref{eq:local1}),
   hence there is a high probability for this flip to occur
   (see equation \eqref{eq:GlauberProbability}),
   and viceversa, if it were to flip from 
   \textsf{-1} to \textsf{+1}, then we would have $\Delta E_x(\omega')=-8$ and so
   transitioning from $\omega'$ to $\omega$ would be unlikely.
   }
   \label{fig:isingGlauber}
\end{figure}
\begin{enumerate}
	\item \label{GD1}
	Choose a site $x \in \mathcal S$ uniformly at random and then compute the local
	energy of $\omega$ at $x$, namely,
	\begin{align} \label{eq:local0}
		E_x(\omega) & \triangleq \sum\limits_{\substack{A\Subset \mathcal S\\ A \ni x}}\Phi_A(\omega) \\ \label{eq:local1}
		& = -h \omega_x - 2J\omega_x \sigma_x(\omega) - J\sum\limits_{y\in N(x)}\sum\limits_{z\in N(y)\setminus\{x\}} \omega_y\omega_z.
	\end{align}
	\item \label{GD2}
	Compute the change in energy $\Delta E_x(\omega)$
	of the configuration $\omega$ if the 
	spin at $x$ were to flip. To be precise, let $\omega' \in \Omega$ be
	defined by the rule $\omega'_y = \omega_y$ if and only if $y \neq x$,
	and then let
	\begin{align}
	\label{eq:gd21}
		\Delta E_x(\omega) & \triangleq E(\omega') - E(\omega) \\
		\label{eq:gd22}
		& = E_x(\omega') - E_x(\omega) \\
		\label{eq:gd23}
		& = 2h \omega_x + 4J\omega_x \sigma_x(\omega).
	\end{align}
	\item \label{GD3}
	Accept the new configuration $\omega'$, i.e. change the state of $\omega$ in site $x$, with probability
	\begin{equation}
	\label{eq:GlauberProbability}
	    \frac{1}{1+ e^{-\beta \cdot \Delta E_x(\omega)}}.
	\end{equation}
	\item
	Repeat \ref{GD1}-\ref{GD3}.
\end{enumerate}

\begin{figure}[t] 
   \centering
   \includegraphics[width=4.7in]{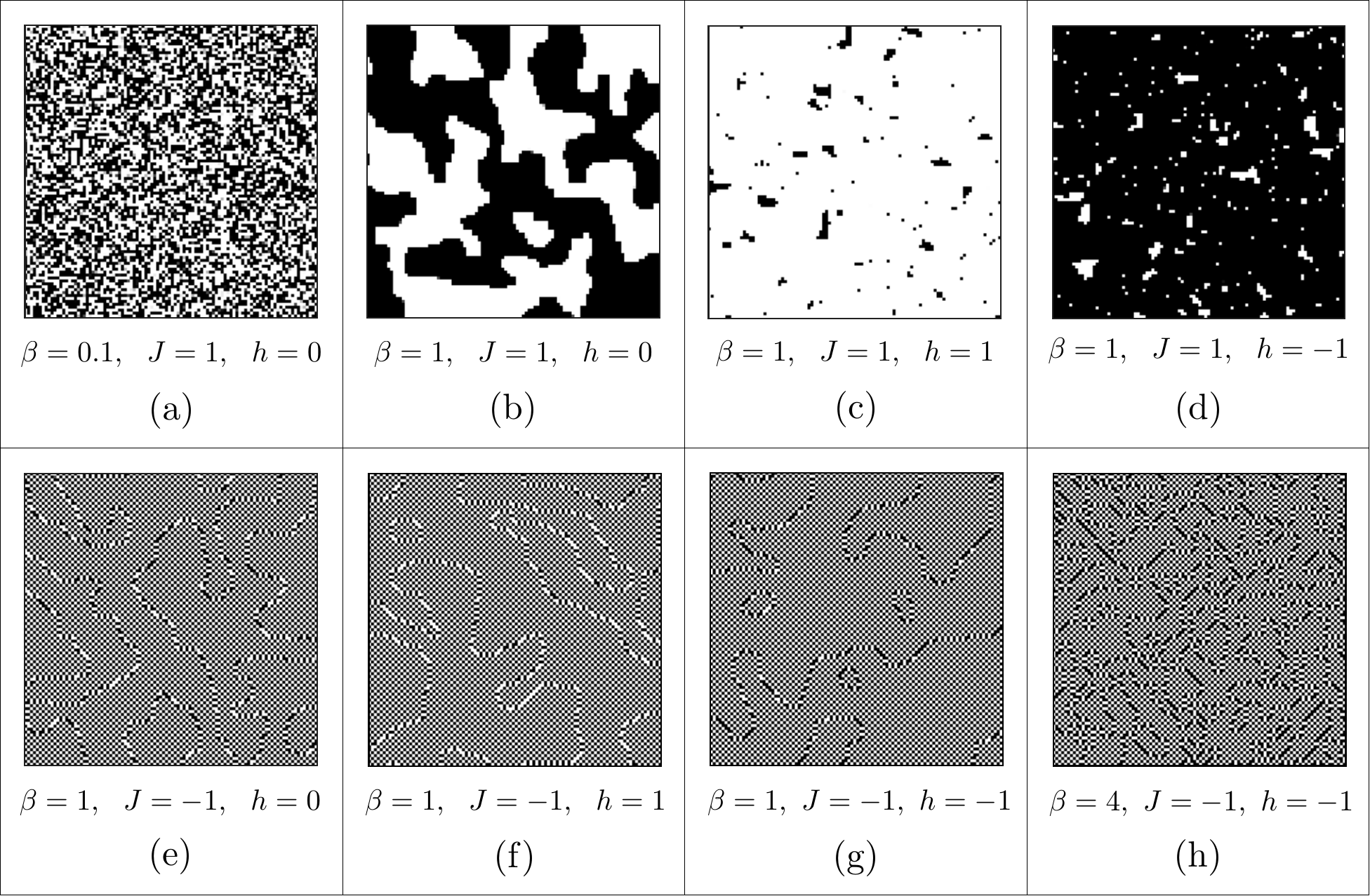} 
   \caption{Simulations of generic configurations for the Ising model for different
   parameter values ($\beta$ is the temperature, $J$ is the weight
   and $h$ is the external magnetic field):
   (a) the temperature is nearly zero,
   so the field behaves like an independent \textsf{Bernoulli}$(1/2)$ random field;
   (b) now the temperature is positive, the system is ferromagnetic
   and the preferred configurations have large connected components of particles
   with the same spin that form ``islands'' of \textsf{+1}s
   (positively charged, in white) and \textsf{-1}s (negatively charged, in black);
   (c) in this case there is an external positive magnetic field that influences
   the energy and so it eventually converges to the configuration
   made out of particles with positive spins;
   (d) now the external magnetic field is negative and the preferred
   configuration is made out of particles with negative spins;
   (e) the system is now antiferromagnetic (with no external magnetic field) and, as a consequence,
   the preferred configurations have large components of particles whose spins make
   checkerboards with a predetermined parity; these configurations have \emph{boundaries}
   which are the regions of adjacency of components of distinct parities
   and where the checkerboard pattern is broken;
   (f) in this case the system is again antiferromagnetic but now there is a positive
   external magnetic field acting on the system, it only influences the 
   boundaries that tend to be positively charged (white);
   (g) in this case the system is again antiferromagnetic but now the external magnetic field
   is negative and so the boundaries tend to be negatively charged (black);
   (h) this is like the previous case, the difference is that now we have raised the temperature
   and the simulation yields convergence to metastable states with large boundaries that are
   negatively charged due to the external magnetic field.}
   \label{fig:Ising}
\end{figure}

Algorithms like the above converge to generic configurations in $\Omega$
with respect to the corresponding Gibbs measure in the Ising model.
In Figure \ref{fig:Ising} we present and describe the outcomes one can obtain
for several parameter values.
Other similar algorithms that also converge to generic configurations
have been used, e.g. the Metropolis-Hastings
algorithm differs from Glauber dynamics in that in step 1, the choice of
the site $x$ is deterministic (for example, the site $x$ is picked one by one,
following some prescribed order in $\mathcal S$), and also in step 3 the acceptance probability
always flips in favor of lowering the energy since a flip always occurs if
$\Delta E_x(\omega) \geq 0$ (here, the temperature is assumed to be positive), 
otherwise the probability of accepting a flip is $e^{\beta \cdot \Delta E/T}$.

These type of models and simulations
admit many variants and have been generalized in many different
directions. For example, they can incorporate
boundary conditions, also called \emph{frozen regions},
as well as hard square type of restrictions.
Our goal is to adapt some of these ideas to the universe of musical scales,
to develop a robust enough model to perform systematic simulations that can yield
subsets of musical scales with precise properties. Let us move now towards this direction.

\section{The \emph{Scaletor} model}
\label{sec:scaletor}

In this section we describe the mathematical aspects in the implementation of
\emph{Scaletor}, the open access application that adapts the ideas of the Ising model
in the context of configurations of musical scales. So the fist subsection
\ref{subsec:configurations} describes
the general space of configurations of musical
scales used in \emph{Scaletor} at any instant of time.
To realize the code, we require a mathematical model to construct all the musical scales, and
for this we follow \cite{Gomez23, Gomez21, GomezNasser21} and continue 
thinking of the class of all musical scales\footnote{We mean all the musical scales \emph{in a given tonality} in the $n$-TET tuning system, for all $n\geq 1$.} as
the combinatorial class $\mathcal C$ of integer compositions. So in the next subsection
\ref{subsec:recursion} we focus on the construction of $\mathcal C$
and also address the modes of scales. In the Ising model there is a network
which corresponds to the square $\mathbb Z^d_n$-lattice,
so we will also require networks of musical scales
and this is precisely the subject of subsection \ref{subsec:networks} (see also \cite{Gomez21}).
Next, in subsection \ref{subsec:interactions} we address the musical interactions that have
been implemented in \emph{Scaletor}: intervalic interactions, balance interactions,
modes interactions and interpolation interactions.
Subsection \ref{subsec:dedicated} aims to describe additional features in \emph{Scaletor}
that in particular allow us to manipulate the spaces of configurations of scales, frozen sites
and the first symbol rule (for the later see \cite{Gomez23, GomezNasser21}).
The last subsection \ref{subsec:GlauberScaletor} addresses the simulation process
in \emph{Scaletor}, in particular we explain how we realize the idea of external magnetic fields.

\subsection{Spaces of configurations of scales}
\label{subsec:configurations}

The starting point is to think of the set of musical scales as the set of sites $\mathcal S$. 
\emph{Scaletor} constructs this class inductively with a recursion
formula (see equation \eqref{eq:recursion} below), it
generates the (finite) sets
$\mathcal C_n\Subset \mathcal C$ of compositions of $n\geq 1$.
Recall that it is precisely $\mathcal C_n$ the set that corresponds to all
the scales in a given tonality in the $n$-TET tuning system, thus
we are naturally interested in the case $n=12$.

Initially, the spin configuration space
is the complete set of configurations
\begin{align}
\label{eq:wholeSystem}
	\Omega = \{\textsf{+1},\textsf{-1}\}^{\mathcal C_n}
\end{align}
(compare with equation \eqref{eq:initialIsing}).
More generally, we will be able to carry out the simulations
on predetermined restricted subsets of sites\footnote{The complement $\mathcal C_n\setminus\mathcal X$ is thought as a set of scales
fixed in the \textsf{Off} state that \emph{do not} interact with the elements of $\mathcal X$.} $\mathcal X \subseteq \mathcal C_n$
with prescribed regions $\mathcal F \subseteq \mathcal X$
that are \emph{frozen} on the positive state and still they \emph{do} interact with the elements of $\mathcal X\setminus\mathcal F$.
To be precise, given $(\mathcal X, \mathcal F)$, the space of configurations
on which the simulations are carried out is
\begin{align}
\label{eq:subsystems}
	\Omega = \Omega (\mathcal X, \mathcal F) \triangleq
	\{ x \in \{\textsf{+1},\textsf{-1}\}^{\mathcal X} \, : \, x_s = \textsf{+1} \, \, \, \forall s \in \mathcal F\}
\end{align}
(for example, equation \eqref{eq:subsystems} yields equation \eqref{eq:wholeSystem}
when $\mathcal X = \mathcal C_n$ and $\mathcal F = \varnothing$).

\subsection{Scales as integer compositions and modes as orbits of cyclic shift actions}
\label{subsec:recursion}

Let us recall formal definitions.

\subsubsection{Integer compositions}
Let $\mathbb N \triangleq \{1,2,3,\ldots \}$
be the class of positive integers.
The class $\mathcal C$ of \emph{integer compositions} consists of
all finite sequences of positive integers, that is, $\mathcal C \triangleq \cup_{k\geq 1} \mathbb N ^k$.
Let $s = (n_1, \ldots , n_k) \in \mathcal C$ be an integer composition.
The \emph{size} of $s$ is $\lvert s \rvert = n \triangleq n_1+\cdots + n_k$ and 
the \emph{length} of $s$ is $\ell(s)\triangleq k$.
Let $\mathcal C_n$ denote the class of integer compositions of size $n \geq 1$.
For example, $\mathcal C_1=\{(1)\}$, $\mathcal C_2=\{(2), (1,1)\}$, $\mathcal C_3=\{(3), (2,1), (1,2), (1,1,1)\}$ and so forth.
Also, for every $k \geq 1$, let $\mathcal C^{(k)} \triangleq \{ s \in \mathcal C \, : \, \ell(s) = k \}$ be the set of integer compositions of length $k$, and for every subset $A \subseteq \mathcal C$, let $A^{(k)}\triangleq A\cap \mathcal C^{(k)}$, e.g. $\mathcal C_n^{(k)}$ denotes the set of compositions of $n$ of length $k$.

Let
$s\star 1 \triangleq (n_1, \ldots , n_{k-1}, n_k + 1)$
and
$s\diamond 1 \triangleq (n_1, \ldots , n_{k-1}, n_k, 1)$.
Also, for every subset $A\subseteq \mathcal C$,
let $A \star 1 \triangleq \{s\star 1 \: \rvert \: s \in S\}$ and
$A \diamond 1 \triangleq \{s\diamond 1 \: \rvert \: s \in S\}$.
Then we have the recursive specification
\begin{equation}
\label{eq:recursion}
	\mathcal C_{n+1} = (\mathcal C_{n} \star 1) + 
	(\mathcal C_n \diamond 1)
\end{equation}
that in words says that any composition of size $n+1$ results from
a composition of size $n$ by either adding a one to its last entry
(in particular, the length remains the same and the size increases by 1) or
appending a new entry with a one at the end of the sequence (in particular,
in this case, both length and size increase one unit each).
Hence there are $\#\mathcal C_n = 2^{n-1}$ compositions of $n$ and
also there are $\#\mathcal C_n^{(k)} = \binom{n-1}{k-1}$ compositions of $n$ of length $k$.
The distribution of $\ell\colon \mathcal C_n\to\{1,\ldots n\}$
as a random parameter on the probability
space $\mathcal C_n$ where all compositions are equally likely
(uniform distribution) is \textsf{Binomial}$(2^{n-1},1/2)$.
In Figure~\ref{fig:gauss} we illustrate the \emph{bell of scales}, i.e
the histogram of scales according to length, that is, according to the number
of pitch classes. The bell of scales serves to visualize the whole set of
musical scales in the $n$-TET tunning system in a given tonality.

\begin{figure}[h] 
   \centering
   \includegraphics[width=4.7in]{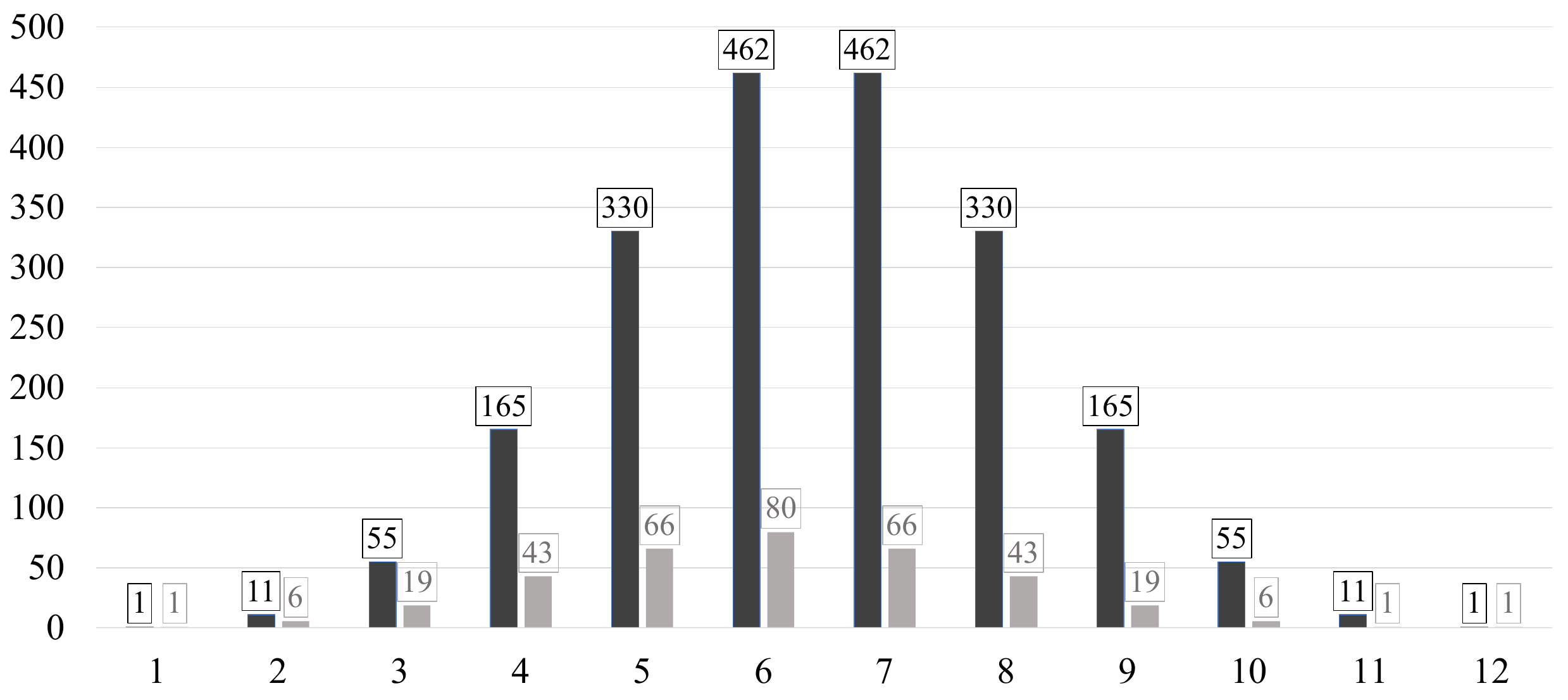} 
   \caption{The dark gray bars illustrate the \emph{bell of scales}: the histogram of scales
   according to the number of pitch classes for the 12-TET tuning system
   (for example, there are 462 scales with either six or seven pitch classes).
   The light gray bars illustrate the histogram of scales modulo the modes: for example,
   there are 80 and 66 mode classes with six and seven pitch classes,
   respectively.}
   \label{fig:gauss}
\end{figure}

From equation \eqref{eq:recursion}, we can define a natural
recursive order in $\mathcal C$: the first element
is $(1)$, then $(2)$ and $(1,1)$, then $(3)$, $(2,1)$, $(1,2)$ and $(1,1,1)$,
and so forth (for example, the compositions of $n$ given by
$(n)$ and $(1, \ldots , 1) $ are always the first and last compositions in
$\mathcal C_n$ and they occupy positions $2^{n-1}$ and $2^n-1$
with respect to this
recursive order, respectively).

The class $\mathcal C_n$ of compositions of $n$ is identified with the
class of all musical scales in the $n$-TET tuning system
in a given tonality like \textsf{C}. 
As in other sources like in \cite{BerlinerCastroMerrittSouthard18},
geometrically, we can illustrate a composition $s \in \mathcal C_n$
of length $\ell(s)=k$ as a rooted $k$-gon that results from a rooted regular $n$-gon
inscribed in a circle after selecting the root
and $k-1$ other vertices separated
according to the composition $s$ (we put the root always on top, at $n$ o'clock,
labelled by the chosen tonality, e.g. we chose \textsf C,
but we may choose any other tonality from $\{
\textsf{C, C\#, D, D\#, E, F, F\#, G, G\#, A, A\#, B}\}$,
see Figure~\ref{fig:scale}. For example, 
$(n)$ and $(1, \ldots , 1)$
are the monotonic and chromatic scales, respectively. 
\emph{Scaletor} uses the recursion in equation \eqref{eq:recursion} to generate all the scales in the $n$-TET tuning system, that is, to generate $\mathcal C_n$.

\begin{figure}[h] 
   \centering
   \includegraphics[width=4.68in]{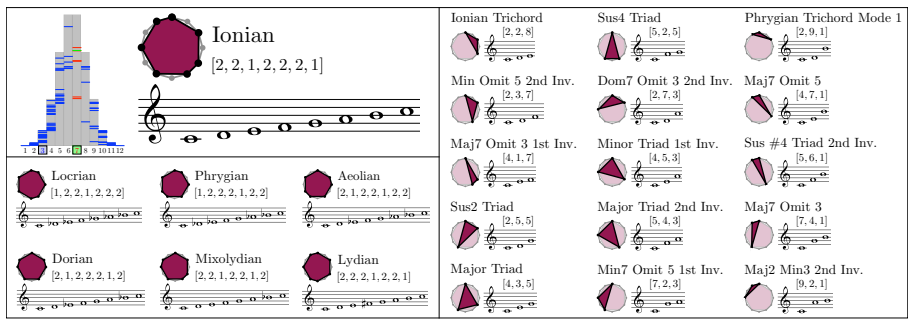}
   \caption{The selected scale in the upper left corner is the Ionian (major) scale (in a given tonality, in this case \textsf{C}), it is represented by the composition $(2,2,1,2,2,2,1)$ of $12$ (it corresponds to the green bar in the bell of scales, in the $7$th column, zoom in). The modes of the Ionian scale are shown in the lower-left corner of the figure (they are represented as red bars in the bell of scales, in the $7$th column too, and in \emph{Scaletor}, by default,
the base note in the modes remains the same for all as shown in this figure,
but pressing the key \fbox{\textsf{\tiny A}} in the keyboard lets you switch back and forth
between this default setting and to show the modes with different base notes in a way that
the pitch classes of the selected scale and all its modes remain the same, e.g. Locrian
would start in \textsf{B}, Mixolydian in \textsf{G}, etc.). The blue bars in the bell of scales
are all the scales that interpolate the Ionian scale (see subsection \ref{subsec:interpolation}).
For example, in the right part of the figure we illustrate all the tri-chords (triangles, in the 3rd column),
containing the root \textsf{C}, that interpolate into the major scale.}
   \label{fig:scale}
\end{figure}

\subsubsection{Modes and cyclic shift actions}

Let $\alpha \colon \mathcal C \to \mathcal C$ be the cyclic shift action
defined for every $s =(n_1, \ldots , n_k)$ by $\alpha(s)=(n_2,\ldots, n_k, n_1)$
(note that $\alpha\rvert_{\mathcal C_n^{(k)}}\colon \mathcal C_n^{(k)} \to \mathcal C_n^{(k)}$).
The elements of the $\alpha$-orbit of $s$, namely $\mathcal O_\alpha(s)\triangleq\{\alpha^n(s) \, : \, n \in \mathbb Z\}$, are the \emph{modes} of $s$.
The size of the orbit of $s$ divides its length: $\#\mathcal O_\alpha(s)\rvert\ell(s)$.
If $A\subseteq \mathcal C_n$ is a subset of scales, a \emph{transversal} is
a set of scales $T=\{s_1,\ldots,s_{t(A)}\}$ such that
for ever $s\in A$ there exists one and only one $j\in\{1,\ldots, t(A)\}$ such
that $s_j\in\mathcal O_\alpha(s)$. Transversals always exist, the
\emph{transversal dimension} $t(A)$ is a well defined positive integer,
and we can always choose $T\subseteq A$.
For more on mathematical aspects of modes of scales see \cite{Gomez23}.

\subsection{Networks of musical scales}
\label{subsec:networks}

We are going to construct various networks with
$\mathcal C$ as the set of vertices. If such a network is denoted by some generic
symbol $\mathfrak G$ and $s \in \mathcal C$, then we  let
the neighborhood of $s$ in $\mathfrak G$ be
$N_{\mathfrak G}(s) \triangleq \{ t \in \mathcal C \, : \, \hbox{$(s,t)$ is an edge in $\mathfrak G$}\}$.

\subsubsection{Composition tree}

For every $n\geq 2$,
we put a thick (resp. thin) edge between elements
$w \in \mathcal C_{n-1}$ and $v \in \mathcal C_{n}$ whenever
$v = w \star 1$ (resp. $v = w \diamond 1$). The \emph{composition tree} $\mathfrak T$
is the resulting network $\mathfrak T$, it has a binary tree with thick and thin edges.
We label the vertices of $\mathfrak T$ with the elements of $\mathcal C$ as
shown in Figure \ref{fig:hypercube} ($\mathfrak T$ is the red binary tree).
In the composition tree, the set of vertices in the $n$th generation
is $\mathcal C_n$. In Figure \ref{fig:hypercube} the vertices are also shown either
as filled or non-filled vertices according to whether the last entry of the composition
is or is not greater than one, respectively.
To be precise, a site $u\in \mathcal C_n$ is filled (resp. non-filled) if it
results from a site $w \in \mathcal C_{n-1}$ through $u =w\star 1$
(resp. $u = w \diamond 1$), i.e. if $u$ is connected in $\mathfrak T$ to its parent
with a thick (resp. thin) edge.

\begin{figure}[h] 
   \centering
   \includegraphics[width=4.7in]{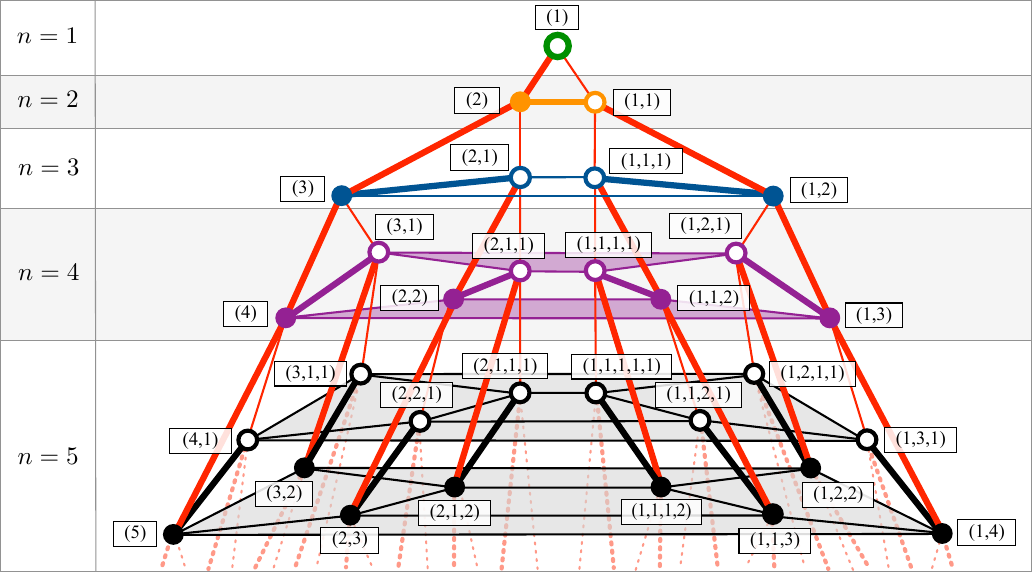} 
   \caption{The beginnings of the composition tree $\mathfrak T$ (in red) and
   the interpolation network $\mathfrak I$. At each stage $n\geq 1$
   the induced interpolation network on $\mathcal C_n$ is a hypercube (the root is green,
   the segment is yellow, the square is blue, the cube is purple and the tesseract is black).
   All the non-filled vertices correspond to compositions that end with $1$.}
   \label{fig:hypercube}
\end{figure}

\subsubsection{Interpolation network}
\label{subsec:interpolation}

With the composition tree $\mathfrak T$,
we can construct the \emph{interpolation network} $\mathfrak I$.
It is also built by adding edges inductively, but in this case, in the $n$th iteration
only edges between certain pair of elements of $\mathcal C_n$ are added,
and again there will bee two types (thick and thin edges).
Initially, there are no edges in $\mathcal C_1$. For $n > 1$,
there is a thick edge in the interpolation network $\mathfrak I$
between $u,v\in \mathcal C_n$ whenever they have
the ``same parent'' in $\mathcal T$,
that is, whenever there exists $w\in \mathcal C_{n-1}$
such that $\{u,v\}=\{w\star 1,w\diamond 1\}$. 
Next, the thin edges in the interpolation network $\mathfrak I$ only
occur between elements in $\mathcal C_n$ of the same 
kind, either filled or non-filled, and the rule is that two vertices
$u,v \in \mathcal C_n$ of the same kind are adjacent in $\mathfrak I$
with a thin edge whenever they have distinct parents in
$\mathfrak T_{n-1} = \mathcal C_{n-1}$
which are adjacent in $\mathfrak I$
(no matter what type of adjacency, either thick or thin).
In words, the thin edges occur between cousins of the
same kind (either both filled, or both non-filled).
This inductive process is implemented in \emph{Scaletor},
let us justify the purpose:

\begin{definition}
Two compositions $s= (s_1, \ldots , s_k) \in \mathcal C_n$
and $t = (t_1, \ldots , t_j)\in \mathcal C_n$
\emph{interpolate} if $k = j+ 1$ (or $j = k + 1$) and
there is $i \in \{1, \dots , k\}$ (or $i \in \{1, \dots , j\}$) such that
$s= (t_1, \ldots , t_{i-1}, s_i, s_{i+1}, t_{i+1}, \ldots , t_{j})$
(or $t= (s_1, \ldots , s_{i-1}, t_i, t_{i+1}, s_{i+1},\ldots , s_{k})$).
\end{definition}

In words, two scales interpolate if one results from the other
by either adding or subtracting a pitch class (distinct from the tonality).
In particular, the lengths of two scales that interpolate most differ exactly by one unit.
The following conclusion easily follows by induction
(we include the proof for completeness, also see \cite{Gomez21}).

\begin{theorem}
\label{thm:interpolation}
Two compositions in $\mathcal C_n$ interpolate if and only if
they are adjacent in the interpolation network. Furthermore, the induced
interpolation network with vertex set $\mathcal C_n$ is a hypercube of dimension $n$.
\end{theorem}

\begin{proof}
We do induction on the size $n \geq 1$. The result is true for $n = 1, 2, 3$ (see Figure~\ref{fig:hypercube}).
Suppose that the result is true for every positive integer less than $n\geq 4$.
Let $s=(s_1,\ldots , s_k)$ and $t=(t_1, \ldots , t_{j})$ be two compositions in $\mathcal C_n$,
and let
$\hat s \in \mathcal C_{n-1}$ and $\hat t \in \mathcal C_{n-1}$ be their
parents, respectively.

Suppose that $s$ and $t$ are adjacent in $\mathfrak I$.
Clearly, if $\hat s = \hat t$ (i.e. if $s$ and $t$ are adjacent in $\mathfrak I$
by a thick edge),
then $s$ and $t$ interpolate.
Suppose that $\hat s \neq \hat t$. Since $s$ and $t$ are adjacent by assumption,
they are of the same kind and $\hat s$ and $\hat t$ are adjacent.
By the inductive hypothesis, $\hat s$ and $\hat t$ interpolate.
If $s$ and $t$ are filled vertices, i.e. if $s_k,t_j >1$, then
$\hat s = (s_1,\ldots , s_k-1)$ and $\hat t = (t_1, \ldots , t_j-1)$,
and since they interpolate, clearly so do $s$ and $t$. Similarly,
if $s$ and $t$ are non-filled vertices, then $s_k=t_j=1$ and hence
$\hat s = (s_1,\ldots , s_{k-1})$ and $\hat t = (t_1, \ldots , t_{j-1})$,
and they interpolate, hence, clearly, so do $s$ and $t$.

Now suppose that $s$ and $t$ interpolate. If $\hat s = \hat t$, then
they are adjacent in $\mathfrak I$. Suppose that $\hat s \neq \hat t$.
First let us show that $s$ and $t$ most be of the same kind. Interpolation
implies $k = j + 1$ or $j = k + 1$. If $s$ and $t$ would be of different kinds,
then we would have $s_k=1$ and $t_j >1$ or viceversa.
Suppose that $s_k=1$ and $t_j>1$. If $k = j + 1$, then we would have
$t_i = s_i$ for all $i < j$ and $t_j = s_{k-1} + 1$, but this implies that
$\hat s = \hat t$, a contradiction.
If $j = k + 1$, then we would have $1 = s_k = t_{k}+t_{k+1}>1$,
again a contradiction. Similarly, supposing that $s_k > 1$ and $t_j = 1$
yields contradictions. Hence $s$ and $t$ are of the same kind.
In any case, whereas
$s_k=t_j = 1$ or $s_k,t_j > 1$, we get that $\hat s$ and $\hat t$ interpolate,
hence the inductive hypothesis implies that they are adjacent, thus so are $s$ and $t$.

The fact that $\mathfrak I$ restricted to $\mathcal C_n$ is a hypercube of dimension $2^{n-1}$ also follows: to obtain $\mathcal C_{n+1}$ we make two copies of $\mathcal C_n$ and pass them
through the composition tree $\mathfrak T$, the binary tree that transform them into $\mathcal C_n \star 1$
and $\mathcal C_n \diamond 1$ (see \eqref{eq:recursion}).
\end{proof}

The previous Theorem \ref{thm:interpolation}, together with
the inductive construction of all musical scales (see equeations \eqref{eq:recursion}),
yield the algorithmic implementation in \emph{Scaletor} that constructs
the integer compositions together with its interpolation network.
Furthermore, we can consider \emph{long range} interpolation networks as follows.
Let $u,v\in\mathcal C_n$ be two scales.
For each $k\geq 1$, let the \emph{$k^+$-step interpolation network}
$\mathfrak I^{k+}$ have a directed edge from $u$ to $v$ if
$u$ results from $v$ by adding $k$ new pitch classes,
and similarly let the \emph{$k^-$-step interpolation network}
$\mathfrak I^{k-}$ have a directed edge from $u$ to $v$
if there is a directed edge from $v$ to $u$ in $\mathfrak I^{k+}$ 
(see Figure \ref{fig:histogramgs}). For example,
we can technically think of the interpolation network $\mathfrak I$
as the union of the $1^+$-step and $1^-$-step interpoletion
(\emph{directed}) networks.
\begin{figure}[h] 
   \centering
   \includegraphics[width=4.7in]{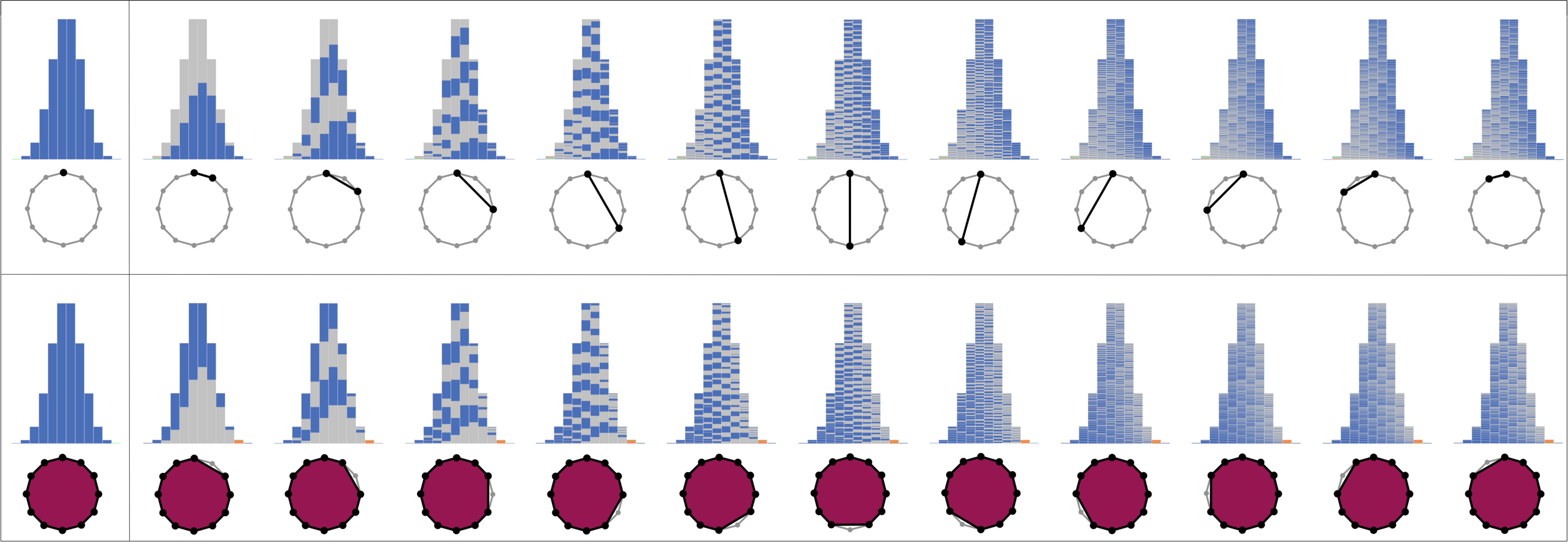}
   \caption{Each bell of scales has a scale under it
   (it corresponds to the green bar in the bell of scales, zoom in).
   The blue bars in each bell of scales are the
   neighbors of the corresponding scale in the
   $k^+$-step and $k^-$-step interpolation networks
   $\mathfrak I^{k+}$ and $\mathfrak I^{k-}$, for all $k\geq 1$. The first row starts with
   the monotonic scale $s=(12)$ (in a given tonality, e.g. \textsf{C}), followed by all the bi-chords
   $s=(1,11), (2,10),\ldots,(11,1)$ (all of them contain the tonal note \textsf{C}).
   The second row are the \emph{complements}, it starts with the
   chromatic scale $s=(1,1,\ldots,1)$, followed by all the eleventh-chords
   $s=(2,1,1,\ldots,1),(1,2,1,\ldots,1), \ldots,(1,\ldots,1,2)$
   (again, all of them contain the tonal note \textsf{C}).
   As before, the red bars represent to the modes of the corresponding scale.
   }
   \label{fig:histogramgs}
\end{figure}

\medskip

\subsubsection{Length network}
\label{sub:sub:lengthNetwork}

The \emph{length network} $\mathfrak L$ also has $\mathcal C$
as the set of vertices and the edges only occur on pairs of elements of $\mathcal C_n$
There is an edge between two sites
$u,v \in \mathcal C_n$ whenever $\ell(u)=\ell(v)$.
Clearly, the length network $\mathfrak L$ restricted to $\mathcal C_n$ is a disjoint union of $n$ cliques $\mathfrak L\rvert_{\mathcal C_n} = \mathfrak K_{m_1} \sqcup \mathfrak K_{m_2} \sqcup \dots \sqcup \mathfrak K_{m_{n}}$
where $\mathfrak K_m$ denotes a clique on $m$ vertices and
$m_k = \binom{n-1}{k-1} = \#\mathcal C_n^{(k)}$
for every $k = 1, \ldots, n$. In the bell of scales,
each column forms the set of vertices of each
of these $n$ cliques.

\subsubsection{Modes network}
\label{sub:sub:modesNetwork}

The \emph{modes network} $\mathfrak M$ is also made out of disjoint unions of cliques:
the clique that contains $s \in \mathcal C$ has vertex set $\mathcal O_\alpha(s)$.
Let $\mathcal T\subset \mathcal C$ be a \emph{transversal} of $\mathcal C/\alpha$,
that is, $\mathcal T$ satisfies that for every $s \in \mathcal C$ there exists a unique $x \in \mathcal T$ such that
$x \in \mathcal O_\alpha(s)$.
Then we have $\mathfrak M = \bigsqcup_{s \in \mathcal T}\mathfrak K_{\#\mathcal O_\alpha(s)}$. 
For example, in the $12$-TET tuning system there are
$\# (\mathcal T \cap \mathcal C_{12}) = 351$ \emph{mode classes} of scales
(i.e. equivalence classes modulo $\alpha$ of size $12$) and, in particular, there are
$\# (\mathcal T \cap \mathcal C^{(5)}_{12}) = 66$
pentatonic mode classes (recall Figure \ref{fig:gauss}).

\subsubsection{Induced networks}

Recall that the spaces of configurations $\Omega = \Omega(\mathcal X, \mathcal F)$
that we will consider are defined, in particular, by a subset $\mathcal X\subseteq \mathcal C_n$
(see equation \eqref{eq:subsystems}).
Here we have defined a set of networks
\begin{align}
\mathfrak N \triangleq \{ \mathcal T, \mathfrak L, \mathfrak M, \mathfrak I, \mathfrak I^{1+}, \mathfrak I^{1-}, \ldots , \mathfrak I^{n+}, \mathfrak I^{n-}\}
\end{align}
and in general, if we are working in a proper restricted subset
$\mathcal X \subsetneq \mathcal C$ in a context in which some networks
$\mathfrak G \in \mathfrak N$ are involved, then we are implicitly assuming
that such $\mathfrak G$s are actually the
\emph{induced} networks with respect to $\mathcal X$.

\subsection{The interactions mix console}
\label{subsec:interactions}

In \emph{Scaletor} we developed 
generic modules or ``tracks'' (see item (b) in Figure \ref{fig:ModulesComponent}),
they are objects that carry components
that serve to control (or mix) the parameters of interactions
of a specific general form discussed below.
A sequence of $K \geq 1$ interactions $\Phi^{(1)}, \Phi^{(2)},\ldots , \Phi^{(K)} \in \Pi(\Omega)$
can be carried in $K$ of these tracks in a container object: the interactions mix console
(see item (c) in Figure \ref{fig:ModulesComponent}).
Each mix will ultimately yield an interaction
$\Phi\in \Pi(\Omega)$, as follows.

First of all, each track carries an
\fbox{\colorbox{OliveGreen}{\textsf{\tiny \color{white}On}}}/\fbox{\colorbox{BrickRed}{\textsf{\tiny \color{Melon}Off}}} (or muting) button that is
modeled with Dirac delta functions, $\delta = \delta_1, \ldots, \delta_K$ (see Figure \ref{fig:deltaButton}).
\begin{figure}[h] 
   \centering
   \includegraphics[width=2in]{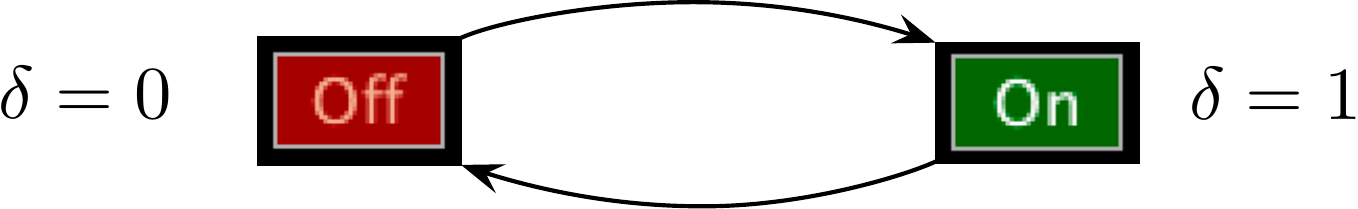} 
   \caption{States of a mute, or
   \fbox{\colorbox{OliveGreen}{\textsf{\tiny \color{white}On}}}/\fbox{\colorbox{BrickRed}{\textsf{\tiny \color{Melon}Off}}}, button:
   it changes its current state when pressed.}
   \label{fig:deltaButton}
\end{figure}

At any instant of time, let
\begin{align}
\Lambda_{\textsf{On}} \triangleq \{ j : \delta_j=1\} \subseteq \{1, \ldots , K \}
\quad \textnormal{ and } \quad
\Lambda_{\textsf{Off}} \triangleq \{1, \ldots , K \} \setminus \Lambda_{\textsf{On}}
\end{align}
be the set of indices of tracks that are turned
\fbox{\colorbox{OliveGreen}{\textsf{\tiny \color{white}On}}}
and \fbox{\colorbox{BrickRed}{\textsf{\tiny \color{Melon}Off}}}, respectively.
Furthermore, the set of tracks $\{ 1, \ldots , K\}$ is also partitioned
into two disjoint sets $\Lambda_{\circ}$ and $\Lambda_{\bullet}$, called
the additive and itersecting tracks.
In \emph{Scaletor}, the final mix of interactions at the level of local energy functions
is 
\begin{align}
\label{eq:mix}
	\fbox{$\displaystyle\Phi_A = \delta \cdot \min \left(\min\left( \Phi^{(k)}_A \right)_{k\in\Lambda_{\bullet} \cap \Lambda_{\textsf{On}}}
	,  \sum\limits_{j\in \Lambda_{\circ}} \delta_j \cdot \Phi_A^{(j)}\right)
	\quad \forall \ A\Subset \mathcal X $}
\end{align}
(above, $\delta$ is yet another Dirac delta function
that acts like the master
\fbox{\colorbox{OliveGreen}{\textsf{\tiny \color{white}On}}}/\fbox{\colorbox{BrickRed}{\textsf{\tiny \color{Melon}Off}}}, 
or muting, switch, recall Figure \ref{fig:deltaButton}).
For example, if all tracks are intersecting and
\fbox{\colorbox{OliveGreen}{\textsf{\tiny \color{white}On}}},
then a configuration $\omega$ that returns a negative weight
for at least one of the $K$ interactions will be penalized and hence
it will tend to be rejected.
The additive part is more subtle, accepting or rejecting depends more on
the average of the additive interactions (like in the Ising model), sort of speak.
To switch the mode of a track, from additive to intersecting or viceversa,
just click the track in a region where there are no buttons nor sliders:
the background of the track when turned
\fbox{\colorbox{OliveGreen}{\textsf{\tiny \color{white}On}}}
will be either green (in the multiplicative case) or
dark red (in the additive case). 
In \emph{Scaletor}, on startup all tracks are multiplicative (green).

According to equation \eqref{eq:interhamilton},
the resulting interaction
$\Phi = \{ \Phi_A \colon \Omega \to \mathbb R \}_{A\Subset \mathcal X}$
defines a hamiltonian energy function $E_\Phi$ which in turn
determines the unique probability measure that maximizes the pressure,
namely the Boltzmann distribution in equation \eqref{eq:boltzman}.
The global inverse temperature, i.e. the parameter $\beta \in \mathbb R$ in equation \eqref{eq:boltzman},
is thought as the ``master volume''.
Item (a) in Figure \ref{fig:ModulesComponent} shows the component to
manipulate $\beta$. Other components shown in this figure are
to manipulate further parameters that we describe next.
\begin{figure}[h] 
   \centering
   \includegraphics[width=4.7in]{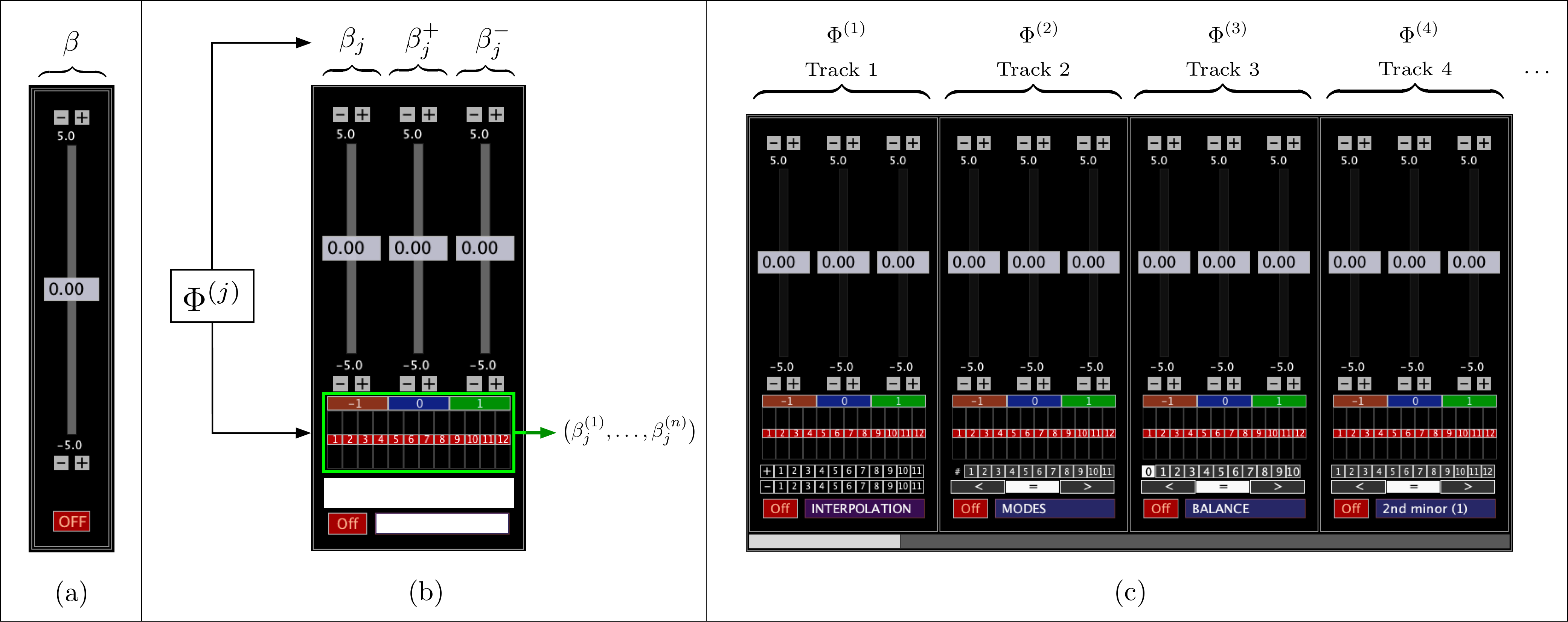} 
   \caption{Thermodynamic components in \emph{Scaletor}: (a) Master volume.
   Component that controls the inverse temperature $\beta$ in equation \eqref{eq:boltzman}.
   The buttons labelled $+$ and $-$ on both the top and the bottom of the slider
   are to increase and decrease the maximum and minimum values that $\beta$
   can take when moving the slider, and by default this interval is $[-5,+5]$.
   (b) Generic interaction track. The three sliders above with
   their dedicated buttons labelled $+$ and $-$ act in a similar way to the
   master volume described in the previous item, but now they control the
   volume of the $j$th channel, namely $\beta_j$, and the corresponding panning
   volumes $\beta_j^+$ and $\beta_j^-$, and below is the
   $n$-band equalizer made out of sliders that
   control the parameters $\beta_j^{(k)}$ (the buttons labelled $+1$, $0$ and $-1$
   above the sliders of the equalizer automatically set all the parameters $\beta_j^{(k)}$
   in the equalizer equal to such a value when pressed).
   The blank areas are meant to provide space
   for possible further controls when the generic module is developed further
   in order to support the interaction that is being designed to be carried on it.
   (c) Interactions mix console. A container component designed to carry sequences
   $\Phi^{(1)},\Phi^{(2)}, \ldots , \Phi^{(K)}$ of
   interaction tracks, each one developed to capture specific properties and
   all with a form given in equation
   \eqref{eq:moduleInteraction}. The scroll bar below serves to move through all the
   $K$ given channels.}
   \label{fig:ModulesComponent}
\end{figure}

To provide flexible and systematic control to
manipulate the simulated annealing processes,
each track that carries one of the $K$ interactions $\Phi^{(j)}$
is equipped with
its own \emph{volume} $\beta_j \in \mathbb R$, 
its two \emph{panning} volumes
$\beta_j^+,\beta_j^- \in \mathbb R$ (they are meant to control how much
a property is ``favored'' or ``penalized'', respectively), and moreover,
the length network $\mathfrak L$ described in subsection
\ref{sub:sub:lengthNetwork} is also implemented
(on each track), modeled as an \emph{$n$-band equalizer}
$\beta_j^{(1)},\beta_j^{(2)},\ldots,\beta_j^{(n)} \in [-1,1]$.
To be precise, for every (finite) region $A \Subset \mathcal X$,
the local energy functions $\Phi_A^{(j)}$ of each of the
interactions $\Phi^{(j)}$ will have the form
\begin{align}
\label{eq:moduleInteraction}
\fbox{$
	\displaystyle\Phi_A^{(j)} = \beta_j \cdot
	\sum\limits_{k=1}^n \beta_j^{(k)} \cdot \left(\beta_j^+\cdot\Phi_{k, A}^{(j+)} - \beta_j^-\cdot\Phi_{k,A}^{(j-)}\right)
	$}
\end{align}
(see items (b) and (c) in Figure \ref{fig:ModulesComponent}).
Combining equations \eqref{eq:moduleInteraction}
with \eqref{eq:mix} and with further custom parameters
will yield an interactions mix console with plenty of control. It is time
to define musical interactions.

\subsection{Interaction tracks}

To define interactions $\Phi^{(j)}$ with local energy functions of the
form described in equation \eqref{eq:moduleInteraction}, we must define
their positive and a negative interaction components $\Phi_{k}^{(j+)}$
and $\Phi_{k}^{(j-)}$ at each length $k = 1,\ldots,n$, respectively.
Equivalently, for each (finite) region $A\Subset \mathcal X$,
we must define the local energy functions of the positive and negative interaction components at length $k$,
\begin{align}
	\Phi_{k,A}^{(j+)}\colon \Omega \to \mathbb R
	\quad \hbox{and} \quad
	\Phi_{k,A}^{(j-)}\colon \Omega \to \mathbb R.
\end{align}
Let us see some explicit examples.
In what follows, $\llbracket \cdot \rrbracket$
denotes Iverson's brackets notation
defined for every proposition $\textnormal{P}$ by
\begin{align}
	\llbracket \textnormal{P} \rrbracket \triangleq \left\{
	\begin{array}{ll}
		1 \quad & \textnormal{if $\textnormal{P}$ is true } \\
		0 & \textnormal{otherwise.}
	\end{array}
	\right.
\end{align}
Also, for every region $A\Subset \mathcal X$ and every configuration $\omega \in \Omega$, let
\begin{align}
\label{eq:magneticWeight}
	\sigma^{\textsf{On}}_A(\omega) \triangleq
	\sum_{y \in A} \llbracket \omega_y = \textsf{+1} \rrbracket .
\end{align}

\subsubsection{Interval tracks}
\label{subsec:intervalTracks}

Let $s\in \mathcal C_n$ be a scale in the $n$-TET tuning system.
We say that an $m$-interval occurs in $s$
if two consecutive pitch classes in $s$ differ by $m$ units of tone,
or equivalently, if we think of $s$ as a composition of $n$, an $m$-interval is
nothing but a summand in $s$ that equals $m$
(for the $12$-TET tuning system, see Table \ref{tab:intervalicD}).

\begin{table}[h]
	\centering
	\begin{tabular}{c|ccccccccccccc}
		$m$ & $1$ & $2$ & $3$ & $4$ & $5$ & $6$ & $7$ & $8$ & $9$ & $10$ & $11$ & $12$\\
		\hline
		Interval & H & W & m3 & M3 & P4
		& TT & P5 & m6 & M6 & m7 & M7 & 8
	\end{tabular} \medskip
	\caption{Intervallic designations in the 12-TET tuning system: H = half tone, W = whole tone, m3 = minor third, M3 = mayor third, P4 = perfect fourth, TT = augmented forth/diminished fifth, P5 = perfect fifth, m6 = minor sixth, M6 = major sixth, m7 = minor seventh, M7 = major seventh, 8 = octave.}
	\label{tab:intervalicD}
\end{table}

An interaction $\Phi \in \Pi(\Omega)$ is intervalic if it returns information
that depends on the intervals. Interval tracks carry intervalic interactions.
For example, consider the combinatorial parameter $\chi_m \colon \mathcal C \to \mathbb N\cup\{0\}$
that returns the number of summands equal to $m\in \mathbb N$, or to be precise, for every
$s = (n_1,\ldots , n_{\ell(s)}) \in \mathcal C$,
$\chi_m(s) \triangleq \# \big\{j \in \{1, \ldots , \ell(s) \} \: \rvert \: n_j = m\big\}$.
In other words, $\chi_m$ returns the number of $m$-intervals in a scale $s$
formed by two consecutive pitch classes in $s$,
e.g. for the Ionian scale $s=(2,2,1,2,2,2,1)$, $\chi_1(s)=2$, $\chi_2(s)=5$, and otherwise
$\chi_m(s)=0$ for $m \geq 3$. With $\chi_m$ we can define intervalic interactions in many ways.
A simple example, the one implemented in 
\emph{Scaletor}, is as follows. Let $\overset R \sim$ be a metasymbol
representing any element of $\{ <, \leq, = ,\neq ,\geq, > \}$.
For each $j \in\{0, \ldots , n\}$, let $\Phi_k$
be the single site interaction defined for every $s\in \mathcal C_n$ by
\begin{align}
	\label{eq:intervalicPlus}
	\Phi_{k,s}^+(\omega) & \triangleq \llbracket \ell(s) =k \rrbracket
	  \cdot \llbracket \chi_m(s) \overset R \sim j \rrbracket \\
	\label{eq:equalityMinus}
        \Phi_{k,s}^-(\omega) & \triangleq \llbracket \ell(s) =k \rrbracket
	  \cdot (1-\llbracket \chi_m(s) \overset R \sim j \rrbracket) .
\end{align}

\emph{Scaletor} possesses twelve interval tracks with intervalic interactions
$\Phi^{(1)}, \ldots , \Phi^{(12)}$
like $\Phi$ above: one for each $m=1, \ldots , n$.
In each track that carries these intervalic interactions, the value of the
parameter $j$ is controlled with specific components\footnote{In the generic module, these specific components
are located in the blank area described before in item (b) of
Figure \ref{fig:ModulesComponent}.} as shown in Figure
\ref{fig:intervalControl}.
\begin{figure}[h] 
   \centering
   \includegraphics[width=4.7in]{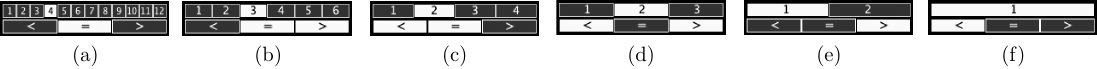} 
   \caption{Components to set the parameter $j$ and the metasymbol $\overset R \sim$.
   The upper part is for the parameter $j$ for (a) H, (b) W, (c) m3, (d) M3, (e) P4, TT, (f) P5, m6, M6, m7, M7, 8. By default, $j = 0$, and when
   a black button labelled with certain value $j$ is pressed,
   then it turns white and the corresponding parameter
   $j$ is set to be such value, e.g. in both (c) and (d) we have $j = 2$ (pressing
   a white button brings the value of $j$ back to zero: $j=0$) but their metasymbols
   are $\leq$ and $\neq$ respectively.}
   \label{fig:intervalControl}
\end{figure}

\subsubsection{Balance tracks}
\label{subsec:balance}
The concept of \emph{balance} of a scale $s\in \mathcal C_n$
is easily defined when a scale of length $k\geq 1$ is pictured
as a polygon that results from choosing $k$ points among $n$
equidistant points arranged in a circle
(the way in which such $k$ point depends on the scale itself),
as described before.
We can suppose that such a circle is unitary and is centered
at the origin. Then the scale is determined by a set of coordinates $\mathbf x^{(1)} = (x_1,y_1)=(0,1), \mathbf x^{(2)} = (x_2,y_2), \ldots, \mathbf x^{(\ell(s))} = (x_{\ell(s)},y_{\ell(s)})$ in the unit circle. Then the \emph{center of balance} of $s$ is defined by
\begin{align}
	\mathbf b(s) & \triangleq\frac{1}{\ell(s)}\sum\limits_{k=1}^{\ell(s)} \mathbf x ^{(k)} 
	= \left(\frac{1}{\ell(s)}\sum\limits_{k=1}^{\ell(s)} x_k ,\frac{1}{\ell(s)}\sum\limits_{k=1}^{\ell(s)} y_k \right) \in [-1,1]^2
\end{align}
and the \emph{balance} of $s$ is $b(s) \triangleq ||\mathbf b(s)||$.
A scale is \emph{balanced} if $b(s)=0$ (equivalently
$\mathbf b (s) = \mathbf 0 \triangleq (0,0)$).
We will define a balance interaction, it is a single site interaction defined for every
$s \in \mathcal C_n$. Once again, we will have a parameter $j =0,1, \ldots , 9, 10$, which is also
manipulated with a dedicated component as shown in Figure \ref{fig:balanceControl}.
\begin{figure}[h] 
   \centering
   \includegraphics[width=1in]{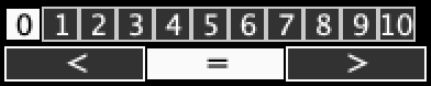} 
   \caption{Dedicated component to manipulate the parameter $j$ in a balance track.}
   \label{fig:balanceControl}
\end{figure}
\begin{itemize}
\item
\fbox{$=$} Equal case. Let 
\begin{align}
	\label{eq:balanceTrackEPlus}
	\Phi_{k,s}^+(\omega) & \triangleq \left\{
	\begin{array}{ll}
	   \llbracket \ell(s) = k \rrbracket \cdot \llbracket \mathbf b(s) = 0 \rrbracket
	   & \textnormal{ if } j=0 \smallskip\\
	    \llbracket \ell(s) = k \rrbracket
            \cdot
	    \llbracket \frac{j-1}{10} < \mathbf b(s) \leq \frac{j}{10} \rrbracket
	   & \textnormal{ if } j>0
	\end{array}
	\right.
	\\
	\label{eq:balanceTrackEMinus}
        \Phi_{k,s}^-(\omega) & \triangleq
        \llbracket  \ell(s) = k \rrbracket \cdot
	  \big(1-\Phi_{k,s}^+(\omega)\big) .
\end{align}
\item 
\fbox{$\geq$} Greater than or equal case. Let 
\begin{align}
	\label{eq:balanceTrackGEPlus}
	\Phi_{k,s}^+(\omega) & \triangleq \left\{
	\begin{array}{ll}
	   \llbracket \ell(s) = k \rrbracket
	   & \textnormal{ if } j=0 \smallskip\\
	    \llbracket \ell(s) = k \rrbracket
            \cdot
	    \llbracket \frac{j-1}{10} < \mathbf b(s) \rrbracket
	   & \textnormal{ if } j>0
	\end{array}
	\right.
	\\
	\label{eq:balanceTrackGEMinus}
        \Phi_{k,s}^-(\omega) & \triangleq
        \llbracket  \ell(s) = k \rrbracket \cdot
	  \big(1-\Phi_{k,s}^+(\omega)\big) .
\end{align}
\item 
\fbox{$>$} Greater than case. Let 
\begin{align}
	\label{eq:balanceTrackGPlus}
	\Phi_{k,s}^+(\omega) & \triangleq
	\llbracket \ell(s) = k \rrbracket
         \cdot
	 \llbracket \frac{j}{10} < \mathbf b(s) \rrbracket
	\\
	\label{eq:balanceTrackGMinus}
        \Phi_{k,s}^-(\omega) & \triangleq
        \llbracket  \ell(s) = k \rrbracket \cdot
	  \big(1-\Phi_{k,s}^+(\omega)\big) .
\end{align}
\item 
\fbox{$\leq$} Less than or equal case. Let 
\begin{align}
	\label{eq:balanceTrackLEPlus}
	\Phi_{k,s}^+(\omega) & \triangleq 
	\llbracket \ell(s) = k \rrbracket
         \cdot
	 \llbracket \mathbf b(s) \leq \frac{j}{10} \rrbracket
	\\
	\label{eq:balanceTrackLEMinus}
        \Phi_{k,s}^-(\omega) & \triangleq
        \llbracket  \ell(s) = k \rrbracket \cdot
	  \big(1-\Phi_{k,s}^+(\omega)\big) .
\end{align}
\item 
\fbox{$<$} Less than case. Let 
\begin{align}
	\label{eq:balanceTrackLPlus}
	\Phi_{k,s}^+(\omega) & \triangleq 
	\llbracket \ell(s) = k \rrbracket
         \cdot
	 \llbracket \mathbf b(s) < \frac{j}{10} \rrbracket
	\\
	\label{eq:balanceTrackLMinus}
        \Phi_{k,s}^-(\omega) & \triangleq
        \llbracket  \ell(s) = k \rrbracket \cdot
	  \big(1-\Phi_{k,s}^+(\omega)\big) .
\end{align}
\end{itemize}
Some of these cases are illustrated in Figure \ref{fig:balance}.
\begin{figure}[htbp] 
   \centering
   \includegraphics[width=4.7in]{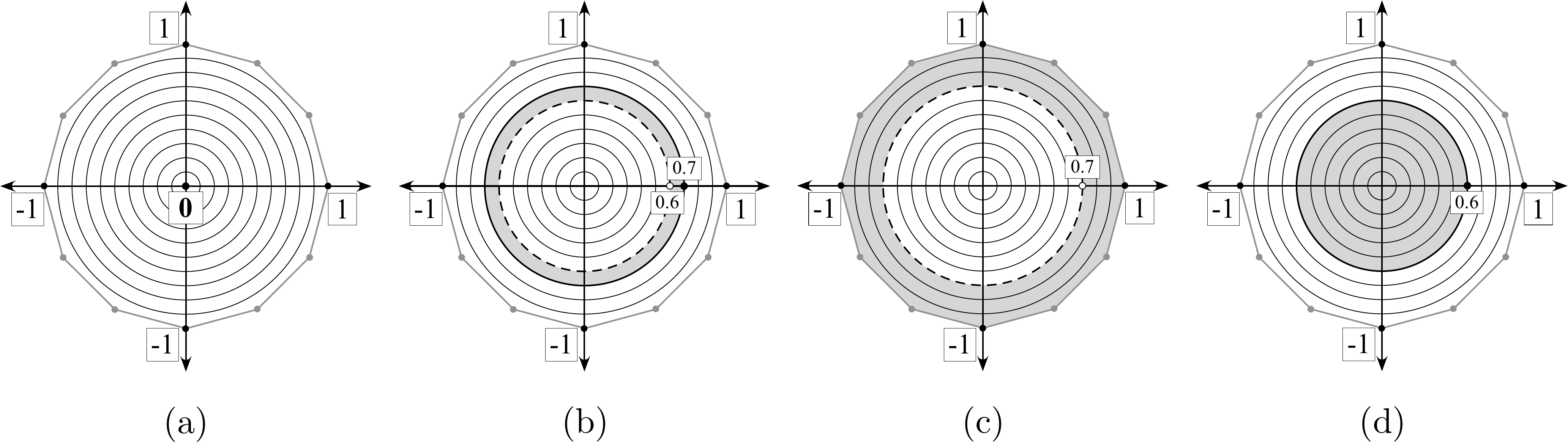} 
   \caption{The balance interaction serves to identify scales with balance within certain ranges (gray areas).
   (a) Here $j = 0$ and the case is \fbox{$=$}.
   (b) Here $j = 7$ and the case is again \fbox{$=$}.
   (c) Here $j = 7$ again, but the case is now \fbox{$>$} (or equivalently, $j=8$ and the case is \fbox{$\geq$}).
   (d) Here $j = 6$ and the case is now \fbox{$\leq$} (or equivalently, $j=7$ and the case is \fbox{$<$}).}
   \label{fig:balance}
\end{figure}

\subsubsection{Modes tracks}
\label{subsec:modesTrack}

Modes tracks are tracks that carry interactions that capture properties of the modes of the scales.
\emph{Scaletor} carries a modes track interaction $\Phi$ that
serves to obtain modes orbitals and transversals of sets of scales
(see \cite{Gomez23}), it is defined as follows. 
The regions $A\Subset \mathcal X$ on which $\Phi_{k,A}$ can be nonzero
are mode classes, that is, regions of the form $A=\{s\} \cup N_{\mathfrak M} (s)$
(note that for these regions, there exists $k = \ell(s)$ such that $A\subseteq \mathcal X^{(k)}$,
that is, $\ell(x) = \ell(y)=k$ for all $x,y \in A$). Henceforth in this subsection $A$ represents
such regions. There is going to be a parameter $j\in \{1, \ldots, n\}$
that controls the size and the number of elements of a mode class in the \textsf{On} state.
\begin{itemize}
\item
\fbox{$=$} Equal case. Let 
\begin{align}
	\label{eq:modesTrackEPlus}
	\Phi_{k,A}^+(\omega) & \triangleq 
	 \llbracket A\subseteq \mathcal X^{(k)} \rrbracket
         \cdot
	 \llbracket \#A = j\rrbracket
	\\
	\label{eq:modesTrackEMinus}
        \Phi_{k,A}^-(\omega) & \triangleq
        \llbracket A \subseteq \mathcal X^{(k)}\rrbracket \cdot
	 \llbracket \#A \neq j \rrbracket .
\end{align}
\item 
\fbox{$\geq$} Greater than or equal case. Let 
\begin{align}
	\label{eq:modesTrackGEPlus}
	\Phi_{k,A}^+(\omega) & \triangleq 
	 \llbracket A\subseteq \mathcal X^{(k)} \rrbracket
         \cdot
         \llbracket \#A \geq j \rrbracket
         \cdot
	 \llbracket \sigma_A^{\textsf{On}}(\omega) \leq j \rrbracket 
	 \\
	\label{eq:modesTrackGEMinus}
        \Phi_{k,A}^-(\omega) & \triangleq
        \llbracket A \subseteq \mathcal X^{(k)}\rrbracket \cdot
	 \llbracket \#A < j \rrbracket.
\end{align}
\item 
\fbox{$>$} Greater than case. Let 
\begin{align}
	\label{eq:modesTrackGPlus}
	\Phi_{k,A}^+(\omega) & \triangleq 
	 \llbracket A\subseteq \mathcal X^{(k)} \rrbracket
         \cdot
         \llbracket \#A \geq j +1 \rrbracket
         \cdot
	 \llbracket \sigma_A^{\textsf{On}}(\omega) \leq j +1 \rrbracket 
	 \\
	\label{eq:modesTrackGMinus}
        \Phi_{k,A}^-(\omega) & \triangleq
        \llbracket A \subseteq \mathcal X^{(k)}\rrbracket \cdot
	 \llbracket \#A \leq j \rrbracket.
\end{align}
\item 
\fbox{$\leq$} Less than or equal case. Let 
\begin{align}
	\label{eq:modesTrackLEMinus}
        \Phi_{k,A}^-(\omega) & \triangleq
        \llbracket A \subseteq \mathcal X^{(k)}\rrbracket \cdot
	 \llbracket \sigma_A^{\textsf{On}}(\omega) \geq j \rrbracket.
\end{align}
\item 
\fbox{$<$} Less than case. Let 
\begin{align}
	\label{eq:modesTrackLEMinus}
        \Phi_{k,A}^-(\omega) & \triangleq
        \llbracket A \subseteq \mathcal X^{(k)}\rrbracket \cdot
	 \llbracket \sigma_A^{\textsf{On}}(\omega) \geq j -1\rrbracket.
\end{align}
\end{itemize}
As before, the parameter $j$ and the corresponding case
is controlled with a custom component like in Figure \ref{fig:modesComponent}.
\begin{figure}[htbp] 
   \centering
   \includegraphics[width=1in]{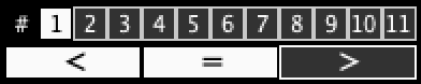} 
   \caption{Component that controls the parameter $j$ in the modes track
   and the corresponding case ($< , \leq, =  , \geq, >$).}
   \label{fig:modesComponent}
\end{figure}

\subsubsection{Interpolation tracks}
\label{subsec:interpolationTrack}

To define an interpolation interaction $\Phi$, we will use
two parameters $J^+,J^-\subseteq \{1,2,\ldots , n -1\}$.
The regions $A\Subset \mathcal X$ on which the local energy functions $\Phi_A$ can
be non-zero will be of the form
\begin{align}
  A = \left(\bigcup\limits_{j\in J^+} N_{\mathfrak I^{j+}}(s)\right)
  \cup
  \left(\bigcup\limits_{j\in J^-} N_{\mathfrak I^{j-}}(s)\right)
  \quad \textnormal{ for some } \quad s\in \mathcal X
\end{align}
and in this case we let
\begin{align}
	\label{eq:interpolTrackPlus}
	\Phi_{k,A}^+(\omega) & \triangleq 
	 \llbracket \ell(s) = k \rrbracket
         \cdot
         \llbracket \sigma_A^{\textsf{On}}(\omega) > 0 \rrbracket
	 \\
	 \label{eq:interpolTrackMinus}
         \Phi_{k,A}^-(\omega) & \triangleq
         \llbracket \ell(s) = k \rrbracket
         \cdot
         \llbracket \sigma_A^{\textsf{On}}(\omega) = 0 \rrbracket
\end{align}
Again, the module that carries the interpolation interaction
is customized with a control component
that is shown in Figure \ref{fig:interpolationControl},
it serves to choose the subsets $J^+,J^-$.
\begin{figure}[h] 
   \centering
   \includegraphics[width=1.2in]{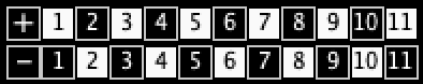} 
   \caption{In this example, $J^+$ and $J^-$ consist of the odd and even numbers, respectively. The top left button labelled $+$ is to set $J^+$ equal to either
   $\{1, \ldots , n-1 \}$ or the empty set, and similarly the bottom left button
   labelled $-$ for $J^-$.}
   \label{fig:interpolationControl}
\end{figure}

This is the last example of interaction tracks that is implemented in \emph{Scaletor}.
Before seeing examples, however, we still need to address a few more features.

\subsection{Dedicated functionalities in \emph{Scaletor}}
\label{subsec:dedicated}

We need to explain a few more functionalities in \emph{Scaletor}:
\begin{itemize}
\item
the MCMC simulation can be both started and stoped;
\item
the processes of setting subsets of restricted sites
and frozen sites in the \textsf{On} state;
\item
the capability of computing scales under the first symbol rule (see
\cite{GomezNasser21, Gomez23}).
\end{itemize}

\subsubsection{Starting and stopping the MCMC simulation machine}

\emph{Scaletor} executes a MCMC
random process that can be started and stoped at any time,
and there is a dedicated button to do this (see Figure \ref{fig:startStopGlauber}).

\begin{figure}[h] 
   \centering
   \includegraphics[width=3in]{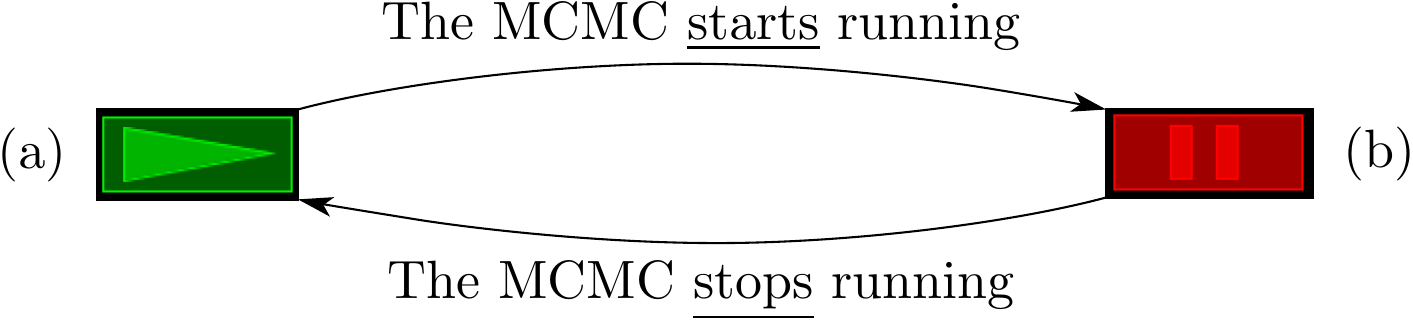} \medskip
   \caption{The two states of the \emph{start}/\emph{stop} button:
   In state (a) the MCMC is not running and there is a fixed configuration $x^{(0)} \in \Omega$.
   When the button is pressed, the button transitions to state (b)
   the MCMC \underline{starts} running
   and a sequence of configurations
   $x^{(0)}, x^{(1)}, \ldots \in \Omega$ is constantly being generated
   according to the probability law of the process that corresponds to
   the given parameters at the moment.
   When pressed again, the button transitions back to (a)
   and the MCMC \underline{stops} running with $x^{(0)} = x^{(N_0)}$, where
   $N_0$ is the index of the last configuration generated by the simulation before the button in state (b) was pressed.}
   \label{fig:startStopGlauber}
\end{figure}

When the MCMC is not running, there is a fixed (initial) configuration $x^{(0)} \in \Omega$.
On startup, the MCMC is not running and this initial configuration is, by default,
$x^{(0)}=\{\textsf{+1}\}^{\mathcal X}=\{\textsf{+1}\}^{\mathcal C_n}$,
that is, all the scales in the $n$-TET tuning system are in the \textsf{On} state.
Otherwise, if the MCMC is running, then a sequence of new configurations
is constantly being generated according to the laws of the underlying
random process at that moment, say $x^{(0)}, x^{(1)}, x^{(2)}, \ldots \in \Omega$
(perhaps the new configurations are eventually
constant, e.g. when the Gibbs measures are Dirac measures and the corresponding
MCMC starts from a configuration that eventually converge to the
configuration on which the point mass limit is based).
When the MCMC stops running, let $x^{(0)} = x^{(N_0)}$ where
$N_0$ is the index of the last configuration generated by the simulation before it stopped.

\subsubsection{Restricted subsets and frozen regions}

Recall that
\emph{Scaletor} allows to restrict to \emph{subsets} $\mathcal X \subseteq \mathcal C_n$
of sites. Also, it admits \emph{frozen} regions $\mathcal F \subseteq \mathcal X$
fixed on the positive state, $\{\textsf{+1}\}^{\mathcal F}$, so that
the final configuration space $\Omega=\Omega(\mathcal X, \mathcal F)$ is as in equation \eqref{eq:subsystems}.
On startup, $\mathcal X = \mathcal C_n$ and $\mathcal F = \varnothing$,
hence $\Omega = \{\textsf{+1},\textsf{-1}\} ^{\mathcal C_n}$
is the full space of configurations supported on all
the scales in the $n$-TET tuning system.
But we could be interested in making simulations on
special subsets $\mathcal X$ of restricted sites, e.g. on scales with $k$ pitch classes,
in which case we could consider, for example, leting $\mathcal X = \mathcal C_{n}^{(k)}$
(this scenario occurs e.g. when searchings trichords and/or pentatonic scales), etc.

As we shall see, allowing restricted subsets and frozen regions of sites 
will be useful. The components that manipulate these two properties
in the GUI of \emph{Scaletor} are shown in Figure \ref{fig:FreezeSub}.
Now we move towards describing the processes of
setting restricted subsets and frozen regions of sites.

\begin{figure}[h] 
   \centering
   \includegraphics[width=4in]{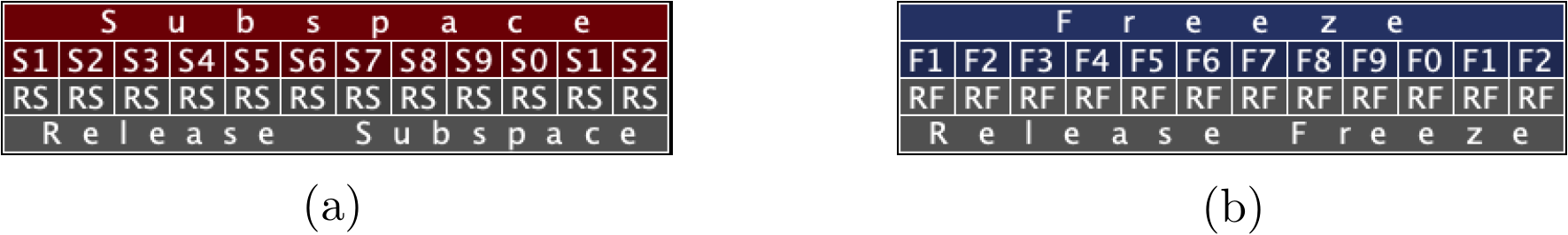}\medskip
   \caption{The two dedicated components in the GUI of \emph{Scaletor}
   to set (a) restricted subsets of sites and (b) frozen regions in specific configurations.}
   \label{fig:FreezeSub}
\end{figure}

\subsubsection{Restricting and releasing restricted subsets of sites}

Restricting sites can be done with the 
\fbox{\colorbox{subspace}{\textsf{\tiny \color{white}Subspace}}} button
(see (a) in Figure \ref{fig:FreezeSub}, it is at the top).
This button can be pressed any time any number of times,
whenever the MCMC simulation is running or not.
It can be pressed even when an (initial)
restricted set of sites
$\mathcal X = \mathcal X^{(0)} \subseteq \mathcal C_{n}$
and/or a frozen region determined by a subset $\mathcal F \subseteq \mathcal X$
have already been established (recall that $\mathcal X = \mathcal C_{n}$
and $\mathcal F=\varnothing$ on startup).
When the \fbox{\colorbox{subspace}{\textsf{\tiny \color{white}Subspace}}} button is pressed,
the following inductive algorithm is executed:
\begin{enumerate}
	\item[0.]
	Start with the current (initial) set of restricted sites
	$\mathcal X^{(0)} = \mathcal X$ and the current
	(initial) configuration $x^{(0)} \in \Omega(\mathcal X^{(0)}, \mathcal F)$.
	\item
	Suppose that $\mathcal X^{(j)}$ and $x^{(j)}$ are given.
	If the MCMC is running, then take one step forward in such a random process
	and get a new configuration
	$x^{(j+1)} \in \Omega \big(\mathcal X^{(j)}, \mathcal F\big)$
	Otherwise, let $x^{(j+1)}\triangleq x^{(j)}$ if the MCMC is not running.
	\item
	Let
	$\mathcal X^{(j+1)} = \{ s \in \mathcal X^{(j)} \, : \, x^{(j+1)}_s = \textsf{+1}\}$.
	Observe that $\mathcal F \subseteq \mathcal X^{(j+1)} \subseteq \mathcal X^{(j)}$.
	\item
	Repeat steps 1 and 2 a finite but arbitrarily large number
	$M_0 \in \mathbb N$ of times,
	with $\mathcal X^{(j+1)}$
	and $x^{(j+1)}\mid_{\mathcal X^{(j+1)}}$
	playing the roles of $\mathcal X^{(j)}$ and $x^{(j)}$, respectively.
\end{enumerate}
If all the parameters are kept fixed while the subspace algorithm is being executed,
then we will get a decreasing random sequence of subsets of sites
$\mathcal X^{(0)} \supseteq \mathcal X^{(1)} \supseteq \ldots$ that
must eventually stabilize, so that if $M_0$ is big enough, then we will have
\begin{align}
\label{eq:speed}
	\mathcal X \triangleq \mathcal X^{(M_0)}=\mathcal X^{(M_0+k)} \, \, \forall k\geq 0.
\end{align}
Observe that when stability is reached after $M_0$ or more iterations,
we will always have $x^{(M_0+k)} = \{\textsf{+1}\}^{\mathcal X}$.
The speed of stabilization, i.e. how big $M_0$ must be in order for equation
\eqref{eq:speed} to hold, depends on
the interactions and the values of their parameters. In \emph{Scaletor}, the default
value for the number of iterations is $M_0=20$, which is
usually enough (press the button several times if necessary).
After pressing this button, in order to avoid awkward outputs,
do not press other buttons for a couple of seconds and let all the iterations finish. 
For example, if the MCMC is running and the interaction $\Phi$ is \emph{trivial}
(i.e. when the mix is empty, in which case the random field consists of
independent \textsf{Bernoulli(1/2)}
random variables), then as a result of pressing the
\fbox{\colorbox{subspace}{\textsf{\tiny \color{white}Subspace}}}
button
(several times if necessary --but rather unlikely--),
we will surely obtain $\mathcal X = \varnothing$.
Thus, in a way, pressing the
\fbox{\colorbox{subspace}{\textsf{\tiny \color{white}Subspace}}}
button has an effect similar to that of a ``stochastic intersection operator'', so to speak.

Below the \fbox{\colorbox{subspace}{\textsf{\tiny \color{white}Subspace}}} button
there are $n$ 
\fbox{\colorbox{subspace}{\textsf{\tiny \color{white}S$k$}}} buttons, with $k =\textsf{1}, \textsf{2}, \ldots , n$
(again see (a) in Figure \ref{fig:FreezeSub}).
These buttons act exactly like the
\fbox{\colorbox{subspace}{\textsf{\tiny \color{white}Subspace}}}
button but they are \emph{local} in the sense that they only affect sites of length $k$.

Below the \fbox{\colorbox{subspace}{\textsf{\tiny \color{white}Subspace}}}
and \fbox{\colorbox{subspace}{\textsf{\tiny \color{white}S$k$}}}
buttons described above, there are the corresponding
\fbox{\colorbox{gray}{\textsf{\tiny \color{white}Release Subspace}}} and $n$
 \fbox{\colorbox{gray}{\textsf{\tiny \color{white}RS}}} buttons,
and they are precisely meant to \emph{release} restricted subset of sites.
There is one \emph{large} \fbox{\colorbox{gray}{\textsf{\tiny \color{white}Release Subspace}}}
button and also there are $n$ \emph{local}  
buttons, the \fbox{\colorbox{gray}{\textsf{\tiny \color{white}RS}}} buttons.
On startup, all these release subspace buttons are gray, as shown
in item (a) of Figure \ref{fig:FreezeSub},
and they turn dark red whenever there are restricted subsets of sites
that can be released with these buttons.
Pressing the \fbox{\colorbox{subspace}{\textsf{\tiny \color{white}Release Subspace}}}
button sets $\mathcal X = \mathcal C_{n}$ (hence the button turns gray).
Similarly, if $\mathcal X$ is the set of restricted sites before pressing the local
\fbox{\colorbox{subspace}{\textsf{\tiny \color{white}RS}}} button below the
\fbox{\colorbox{subspace}{\textsf{\tiny \color{white}S$k$}}} button, then, after pressing
such a \emph{local} releasing button, 
the new set of restricted sites is 
$\mathcal X \cup \mathcal C_{n}^{(k)}$ (and thus, similarly, the button turns gray).

\subsubsection{Freezing and releasing frozen regions}
Freezing specific sites in the positive state \textsf{On} 
can be done with the large
\fbox{\colorbox{freeze}{\textsf{\tiny \color{white}Freeze}}}
button (in blue) at the top 
(see again Figure \ref{fig:FreezeSub}).
It acts in a very similar way as the Subspace buttons described before.
It can be activated any time, even when a subset of (initial) frozen sites
$\mathcal F = \mathcal F^{(0)} \subseteq \mathcal X$
has already been established (recall again that
on startup, $\mathcal F = \varnothing$).
When the \fbox{\colorbox{freeze}{\textsf{\tiny \color{white}Freeze}}}
button is pressed, the following inductive algorithm is executed:
\begin{enumerate}
	\item[0.]
	Start with the current (initial) configuration
	$x^{(0)} \in \Omega(\mathcal X,\mathcal F^{(0)})$:
	\item
	Suppose that $\mathcal F^{(j)}$ and $x^{(j)}$ are given.
	If the MCMC is running, then take one step
	forward in such a random process,
	and, as a result, get a new configuration
	$x^{(j+1)} \in \Omega = \Omega \big(\mathcal X, \mathcal F ^{(j)}\big)$.
	Otherwise, let
	$x^{(j+1)}\triangleq x^{(j)}$.
	\item
	Let $\mathcal F^{(j+1)} = \{ s \in \mathcal X \, : \, x^{(j+1)}_s = \textsf{+1} \}$.
	Observe that $\mathcal F^{(j+1)} \supseteq \mathcal F^{(j)}$.
	\item
	Repeat steps 1 and 2
	a finite but arbitrarily large number $M_0 \in \mathbb N$ of times,
	with $\mathcal F^{(j+1)}$
	and $x^{(j+1)}$
	playing the roles of $\mathcal F^{(j)}$ and $x^{(j)}$, respectively.
\end{enumerate}
If all the parameters are kept fixed while the freezing algorithm is being executed,
then we will get an increasing random sequence of subsets of sites
$\mathcal F^{(0)} \subseteq \mathcal F^{(1)} \subseteq \ldots$
that must stabilize, and again the default value $M_0=20$ is
enough for practical purposes. For example,
if the MCMC is running and the interaction
$\Phi$ is trivial (i.e the random field is formed by
i.i.d. \textnormal{Bernoulli}$(1/2)$ random variables), then, as a result of pressing the
\fbox{\colorbox{freeze}{\textsf{\tiny \color{white}Freeze}}} button
(several times if really necessary --a rather unlikely event--),
we will surely get $\mathcal F = \mathcal X$, i.e.
the frozen region will consist of all the scales in the given restricted
subset of sites $\mathcal X$.
Thus, in a way, pressing the
\fbox{\colorbox{freeze}{\textsf{\tiny \color{white}Freeze}}}
button has an effect similar to that of a ``stochastic union operator'', so to speak.

Below the large \fbox{\colorbox{freeze}{\textsf{\tiny \color{white}Freeze}}}
button there are $n$ 
\fbox{\colorbox{freeze}{\textsf{\tiny \color{white}F$k$}}} buttons, with $k =1, 2, \ldots , n$.
These buttons act exactly like the
\fbox{\colorbox{freeze}{\textsf{\tiny \color{white}Freeze}}}
button but their actions are restricted to sites of $\mathcal X$ of length $k$.

Also, just like the buttons to Release Subspaces of restricted sites,
below the blue freezing buttons described above, there are the corresponding
buttons that release, or \emph{unfreeze}, frozen sites. Similarly again, there are
one large and $n$ local buttons to unfreeze frozen sites.
On startup they are all gray and they turn blue
whenever there are frozen sites that can be unfrozen by these buttons.
The large \fbox{\colorbox{gray}{\textsf{\tiny \color{white}Release Freeze}}}
button makes $\mathcal F = \varnothing$ when pressed, i.e. it
unfreezes any frozen site in $\mathcal X$. Similarly,
the local \fbox{\colorbox{gray}{\textsf{\tiny \color{white}RF}}} buttons
below each
\fbox{\colorbox{freeze}{\textsf{\tiny \color{white}F$k$}}}
button with $k =$ \textsf{1}, \textsf{2}, $\ldots , n$, when pressed, it
unfreezes any frozen site in $\mathcal X \cap \mathcal C_{n}^{(k)}$,
i.e. if $\mathcal F$ is the set of frozen sites before pressing such a local button,
then, after pressing it, 
the new set of frozen sites is $\mathcal F \setminus \mathcal C_{n}^{(k)}$
(the release buttons are also gray, and they turn dark blue when
there are frozen sites that can be released).

In addition, in \emph{Scaletor} there is always a selected scale,
and pressing the \fbox{\textsf{F}} key in the keyboard Freezes the selected
scale in the \textsf{On} state, that is, given $\mathcal X$, $\mathcal F$, and $x\in \mathcal X$, then,
as a result of pressing the \fbox{\textsf{F}} key in the keyboard,
we get the new frozen region $\mathcal F \cup \{x\}$ in the \textsf{On} state.
This will be useful. Also, changing the selected scale can be
done either by clicking on a scale in the bell of scales or with the arrow keys
(\fbox{$\uparrow$}, \fbox{$\rightarrow$}, \fbox{$\downarrow$}, \fbox{$\leftarrow$})
that move a (green) cursor along the scales in the \textsf{On} state.

\subsubsection{First symbol rule}
\label{subsec:fsr}
In \cite{Gomez23, Gomez21, GomezNasser21} a method to produce
scales from a symbolic sequences was introduced and studied.
In \emph{Scaletor} there is a dedicated button that does this job.
More precisely, whenever the
\fbox{\colorbox{fsr}{\textsf{\tiny \color{white}First Symbol Rule}}} button is pressed,
\emph{Scaletor} reads the text file\footnote{This file \emph{must} have this name by default, and it most be saved in the same folder in which \emph{Scaletor} is located.}
\texttt{sequence.txt} as a finite sequence, or word,
of length $N\gg 0$ over an alphabet $\mathcal A$,
say $w\in \mathcal A^N$, and returns the corresponding set of scales
that results from the sequence under the \emph{first symbol rule}:

\begin{itemize}
\item \textsc{First Symbol Rule.}
Let $w =w_1\ldots w_{N}\in \mathcal A^N$. For each $k_0 \geq 1$ such that
$k_0 \leq N - n + 1$, consider the factor $w_{[k_0,k_0+n)} \triangleq w_{k_0} \ldots w_{k_0+n-1}$.
Let $\mathfrak s \triangleq w_{k_0} \in \mathcal A$ and then recursively define
$k_m = \min \{ k>k_{m-1} \ : \ w_{k_m}=\mathfrak s\}$
for all $m\geq 1$ (eventually $k_m=\infty$, when there is no $k$ that satisfies the condition).
If $k_1 - k_0 \geq n$, then the scale that the factor $w_{[k_0,k_0+n)}$ induces
under the first symbol rule is
$(n)$. Otherwise, if $l \geq 1$ is the greatest index such $k_l < k_0+n$, then the
scale that the factor $w_{[k_0,k_0+n)}$ induces under the first symbol rule is
\begin{align}
(k_1-k_0, \ldots, k_{l} - k_{l-1}, n + k_0 - k_l).
\end{align}
\end{itemize}
When the \fbox{\colorbox{fsr}{\textsf{\tiny \color{white}First Symbol Rule}}}
is pressed, the resulting scales will be in the \textsc{On} state and all other
scales will be in the \textsf{Off} state, also the MCMC automatically stops running.
Furthermore, only scales in restricted subsets of sites are considered, and
if there are frozen region of sites in the \textsf{On} state, they will remain
\textsf{On}, regardless of whether or not they are generated by the given
sequence through the first symbol rule.

\subsection{Glauber dynamics in \emph{Scaletor}}
\label{subsec:GlauberScaletor}

The simulation in \emph{Scaletor} is essentially the same as Glauber dynamics in the Ising model
as described in subsection \ref{subsec:GlauberMCMC}.
To be precise, the following three steps refer to the three
steps described in the aforementioned subsection \ref{subsec:GlauberMCMC}.
\begin{enumerate}
\item We start with step \ref{GD1} and its equation \eqref{eq:local0},
but certainly we use the hamiltonian for the \emph{Scaletor} model instead
of the hamiltonian for the Ising model that yielded equation \eqref{eq:local1}).
\item
Step \ref{GD2} is also carried out similarly, with its equations
\eqref{eq:gd21} and \eqref{eq:gd22}.
\item
The final step \ref{GD3} is identical.
\end{enumerate}

Furthermore, \emph{Scaletor} possesses an additional simulation.

\subsubsection{External magnetic fields}

There is a useful alternative simulation that turns frozen regions of sites
into external magnetic fields sort of speak. This has effect only on the interactions
that involve the definition of $\sigma_A^{\textsf{On}}$ (see equation \eqref{eq:magneticWeight}). 
When the external magnetic field has been activated, \emph{Scaletor} uses instead
\begin{align}
	\sigma^{\textsf{On}}_A(\omega) \triangleq
	\sum_{y \in A\cap \mathcal F} \llbracket \omega_y = \textsf{+1} \rrbracket 
	= \mid A\cap \mathcal F \mid.
\end{align}
To activate or deactivate the external magnetic field, press the \fbox{\tiny \textsf{M}} key
(on startup the magnetic field is deactivated).

\section{Examples}
\label{sec:examples}

Let us start with examples that have
one single track that is turned \fbox{\colorbox{OliveGreen}{\textsf{\tiny \color{white}On}}}
and that carries single site interactions.
First we look at interaction tracks from subsection \ref{subsec:intervalTracks}.

\begin{example}
\label{ex:basic5}
{\textsc{Single interval track.}}
Consider the interactions mix console that 
consists of one single track, i.e. there is only one track in the \fbox{\colorbox{OliveGreen}{\textsf{\tiny \color{white}On}}} position,
and it carries the intervalic interaction $\Phi^{(k)} = \Phi$ obtained from equation \eqref{eq:moduleInteraction}
using the interaction components defined by equations \eqref{eq:intervalicPlus} and
\eqref{eq:equalityMinus}.
For a mix where $\delta \neq 0$ (i.e. the master volume is \fbox{\colorbox{OliveGreen}{\textsf{\tiny \color{white}On}}}),
$\beta \gg 0$ (i.e. bring the slider of the master volume all the way up),
and $\beta_k, \beta_k^{+}, \beta_k^{-},\beta_k^{(1)}, \ldots , \beta_k^{(n)} \gg 0$
(i.e. bring the sliders of both the three volumes of the track and the equalizer all the way up),
if the MCMC machine is running, then, for example, the following holds:
\begin{itemize}
\small
\item
In the equality case, the Gibbs measure is the point mass
measure based on the configuration of scales
in $\Omega$ that has in the \textsf{On} state precisely the set of scales in $\Omega$ with exactly
$j$ occurrences of an $m$-interval. For example, in the 12-TET tuning system,
if $m=2$ and $j = 5$, then we obtain the configuration $\omega \in \Omega$
that has in the \textsf{On} position precisely the 21 scales with exactly
five whole tones:  the Ionian scale $(2,2,1,2,2,2,1)$ and its modes (Locrian, Phrygian, Aeolian, Dorian, Mixolydian and Lydian, see again the left part in Figure \ref{fig:scale}),
and another two mode orbitals of size 7,
one is from the Melodic Minor scale $(2,1,2,2,2,2,1)$
and the other one from the Neapolitan scale $(1,2,2,2,2,2,1)$
(see Figure \ref{fig:MelodicMinorNeapolitan}).
\begin{figure}[htbp] 
   \centering
   \includegraphics[width=4.7in]{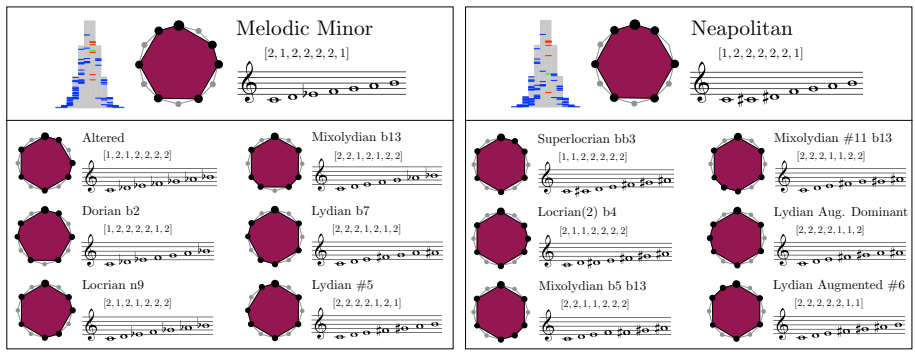} 
   \caption{The orbitals of the Melodic Minor scale and the Neapolitan scale. Together with
   the orbital of the Ionian scale (see the left part in Figure \ref{fig:scale}), we obtain all the scales with exactly 5 second majors formed with consecutive pitch classes.}
   \label{fig:MelodicMinorNeapolitan}
\end{figure}
\item
In the greater than or equal case, again with $m=2$ and $j = 5$,
we get 22 scales with at least five whole tones: the previous 21 from the equality case (length 7),
and the Hexatonic scale $(2,2,2,2,2,2)$
(it has length 6, it corresponds to the chromatic scale in the $6$-TET tuning system,
it is shown in Figure \ref{fig:ModesOne}).
\end{itemize}
\end{example}

Next, the balance interaction from subsection \ref{subsec:balance} is also a single site interaction.

\begin{example}
\label{ex:balanceSingle}
{\textsc{Single balance track.}}
Consider the interactions mix console that again
consists of one single track
and it carries the balance interaction $\Phi^{(k)} = \Phi$ described in subsection
\ref{subsec:balance} (e.g. for the equality case, $\Phi^{(k)}$ is defined through equations \eqref{eq:balanceTrackEPlus} and \eqref{eq:balanceTrackEMinus}).
Again let $\delta \neq 0$, $\beta \gg 0$, and
$\beta_k, \beta_k^{+}, \beta_k^{-},\beta_k^{(1)}, \ldots , \beta_k^{(n)} \gg 0$.
If the MCMC machine is running, then, for the equality case, the Gibbs measure
is a Dirac measure supported on the set of scales that satisfy $b(s) = 0$ if $j=0$,
and $\frac{j-1}{10} < b(s) \leq \frac{j}{10}$ if $j>0$. In Figure \ref{fig:BalanceHistogram}
the dark gray bars form the corresponding histogram of scales according to their
balance that is partition first for perfect balance and then 
through these ten disjoint uniform consecutive intervals
$(0,\frac{1}{10}], (\frac{1}{10}, \frac{2}{10}],\ldots, (\frac{9}{10}, 1]$.
The light gray bars form the histogram modulo modes (e.g. there are
only 18 different shapes formed from a subset of a regular 12-gon
with perfect balance, see Figure \ref{fig:PerfectBalance}).
\begin{figure}[htbp] 
   \centering
   \includegraphics[width=4.7in]{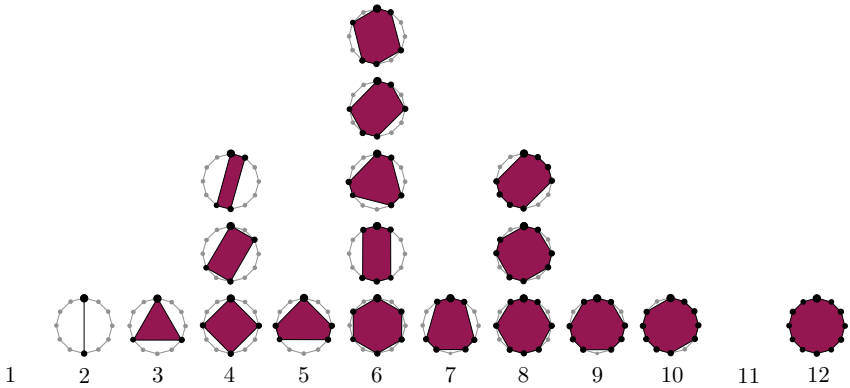} 
   \caption{The 18 shapes with perfect balance for the 12-TET tuning system.}
   \label{fig:PerfectBalance}
\end{figure}
\begin{figure}[h] 
   \centering
   \includegraphics[width=4.7in]{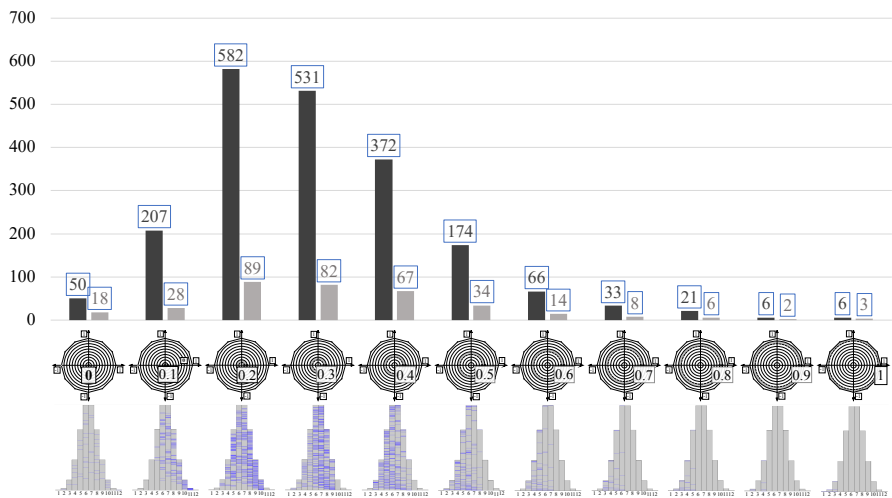} 
   \caption{The dark gray bars correspond to the set of scales that satisfy
   $b(s) = 0$ if $j=0$, and $\frac{j-1}{10} < b(s) \leq \frac{j}{10}$ if $j=1, \ldots , 10$.
   For example, there are exactly 50 scales with perfect balance.
   The light gray bars correspond to the modes classes. At the bottom we see their
   corresponding location in the bell of scales (the blue bars). The histogram with
   light gray bars is for the mode classes, e.g. a modes transversal of
   the set of scales with perfect balance always has cardinality 18.
   }
   \label{fig:BalanceHistogram}
\end{figure}
\end{example}

In the next three examples we explore some basic parameter mixing and
features like the \fbox{\colorbox{subspace}{\textsf{\tiny \color{white}Subspace}}}
and \fbox{\colorbox{freeze}{\textsf{\tiny \color{white}Freeze}}} buttons.

\begin{example}
\label{ex:basic5Comp}
\textsc{Complements.}
Start from the Example \ref{ex:basic5}, with $m=2$ and $j = 5$.
\begin{itemize}
\small
\item
Now let $-\beta \gg 0$ (i.e. bring the slider of the master volume from all the way up to all the way down). Then
we get the complement, that is, in the equality case we get the scales that do
not have exactly five whole tones, and in the greater than or equal case we get
the scales that have less than five whole tones.
\end{itemize}
The same occurs of course if instead
we let either one of the following:
\begin{itemize}
\small
  \item $-\beta_k \gg 0$;
  \item both $-\beta_k^{+}  \gg 0$ and $ -\beta_k^{-} \gg 0$;
  \item $-\beta_k^{(j)}  \gg 0$ for all $j=1,\ldots , n$.
\end{itemize}
Furthermore, the following occurs:
\begin{itemize}
\small
\item
Again if we start from all the sliders all the way up but now we let
$-\beta_k^{(j)}  \gg 0$ only for some $j\in\{1,\ldots , n\}$, then, in the
$j$th column of the bell of scales we will see the complement
restricted to compositions of $n$ of length $j$.
\end{itemize}
Lastly, consider the following:
\begin{itemize}
\small
\item
In the less than case, again if we start from all the sliders all the way up
but now we let $-\beta \gg 0$ (or anyone of the previous items), then
we get the configuration of scales in the second item of example \ref{ex:basic5}.
\end{itemize}
Similar complements can be obtain with the balance track from example \ref{ex:balanceSingle}.
\end{example}

\begin{example}
\label{ex:Subspace}
\textsc{The} \fbox{\colorbox{subspace}{\tiny \color{white}\textsf{Subspace}}} \textsc{button
and restricted subsets of sites.}
Consider again Example \ref{ex:basic5}, but now let $\beta_1^{-} = 0$.
In this case the Gibbs measure is no longer a Dirac measure. Nevertheless,
the preferred configurations are those that always have in the \textsf{On}
state the scales with exactly (in the equality case) 5 whole tones, but not
satisfying this property is not penalized because $\beta_1^{-} = 0$,
so the rest of the scales are distributed like \textsf{Bernoulli}(1/2). In this case we can
observe how the \fbox{\colorbox{subspace}{\textsf{\tiny \color{white}Subspace}}}
button works while the MCMC machine is running:
if pressed, the decreasing sequence of restricted subset of sites
$\mathcal X^{(j)}$ will converge to the set of scales with exactly
(or with at least, in the less than case) 5 whole tones, that is,
to the set of scales in the \textsf{On}
position in the configuration $\omega$ from example \ref{ex:basic5}.
\end{example}

\begin{example}
\label{ex:Freeze}
\textsc{The} \fbox{\colorbox{freeze}{\textsf{\tiny \color{white}Freeze}}} \textsc{button
and frozen regions in the \textsf{On} state.}
Consider again Example \ref{ex:basic5}, again let $\beta_1^{-} = 0$
as in Example \ref{ex:Subspace}, but now also let $-\beta \gg 0$.
Again in this case the Gibbs measure is not a Dirac measure. 
The preferred configurations are those that always have in the \textsf{Off}
state the scales with exactly (or with at least) 5 whole tones, and each of the rest
of the scales are distributed like \textsf{Bernoulli}(1/2). So now we can see how
the \fbox{\colorbox{freeze}{\textsf{\tiny \color{white}Freeze}}} button works while the MCMC is running: if pressed,
then the increasing sequence of frozen regions $\mathcal F^{(j)}$
will converge to the complement described in Example \ref{ex:basic5Comp},
that is to the set of scales in the \textsf{Off}
position in the configuration $\omega$ from Example \ref{ex:basic5}.
\end{example}

Now we continue with the other interactions we have defined,
they are not single site but ‘‘single region'' interactions sort of speak
(see equation \eqref{eq:local0}). Let us start with the modes.

\begin{example}
\label{ex:modesSingle}
{\textsc{Single modes track.}}
Consider the interactions mix console that carries the modes interaction
$\Phi$ defined in subsection \ref{subsec:modesTrack}. Again turn it
\fbox{\colorbox{OliveGreen}{\textsf{\tiny \color{white}On}}} together with the master volume
and set all parameters $\gg 0$ by bringing all the corresponding sliders
all the way up. The following are examples of what can you get:
\begin{itemize}
\small
\item
In the equality case (equations \eqref{eq:modesTrackEPlus}
and \eqref{eq:modesTrackEPlus}), we get Dirac measures for each
$j = 1,\ldots,n$. The support is precisely the scales with exactly $j$ modes.
For example, if $j=1$ then we get the 6 scales with exactly one mode: $(12)$, $(6,6)$,
$(4,4,4)$, $(3,3,3,3)$, $(2,2,2,2,2,2)$, and $(1,1,1,1,1,1,1,1,1,1,1,1)$,
that is, the chromatic scales in the $k$-TET tuning systems for
all $k \mid n$ (see Figure \ref{fig:ModesOne}).
\begin{figure}[htbp] 
   \centering
   \includegraphics[width=4.7in]{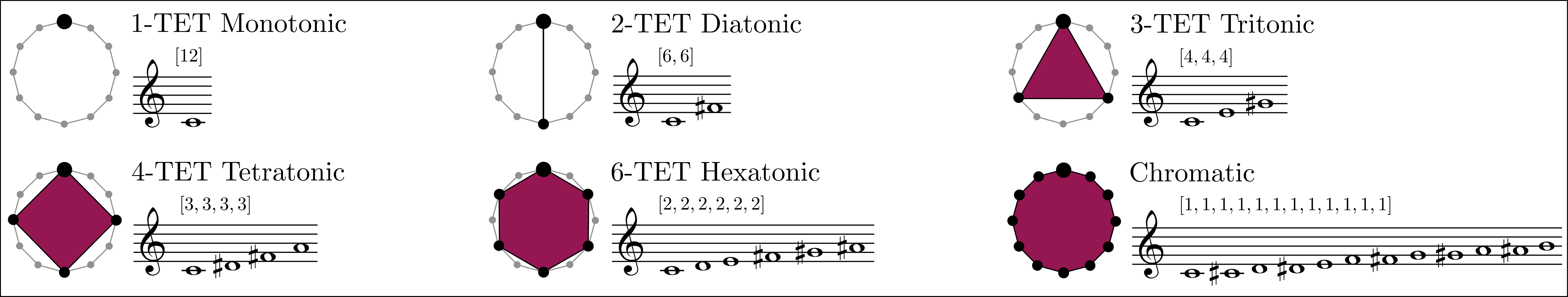} 
   \caption{The six scales with exactly one mode. All but the first have perfect balance.}
   \label{fig:ModesOne}
\end{figure}
\item
Another instance of the equality case for example, we know that
there are $462$ scales of length $6$ (e.g. recall Figure \ref{fig:gauss}),
but if we let $j=6$ we see that there are $450$ scales with exactly
$6$ modes, hence there are $12$ scales of length $6$ that have
less than six modes, so they must have $j$ modes only for $j=1,2,3$,
and we have quick access with the modes track. Let $j=3$ and get
the orbitals of three scales: the Messiaen V scale $(1,1,4,1,1,4)$,
the Messiaen II Truncated scale $(1,2,3,1,2,3)$ and the Raga Indupriya India scale
$(1,3,2,1,3,2)$, for $j=2$ we get the orbital of the Raga Vasanta Pentachord 5aug scale $(1,3,1,3,1,3)$,
and certainly for $j=1$ we get the 6-TET Hexatonic scale $(2,2,2,2,2,2)$.
See Figure \ref{fig:ModesLess6}.
\begin{figure}[htbp] 
   \centering
   \includegraphics[width=4.7in]{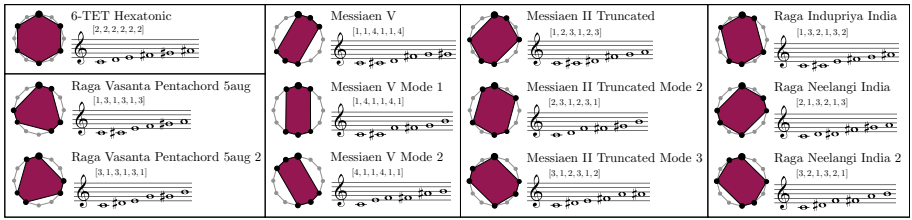} 
   \caption{Orbitals of scales of length 6 with less than 6 elements.}
   \label{fig:ModesLess6}
\end{figure}
\item
\label{eq:modesTransversal}
In the less than or equal case, the orbital of a scale is allowed to have
at most one element in the \textsf{On} state, so the Gibbs measure is not a
Dirac measure, configurations in which there are two or more elements of
an orbital in the \textsf{On} state never occur. Pressing the
the \fbox{\colorbox{freeze}{\textsf{\tiny \color{white}Freeze}}} button will yield a
modes transversal of the scales (see again Figure \ref{fig:gauss},
the histogram with light gray bars), the modes transversal dimension
of the set of scales is 351.
\end{itemize}
\end{example}

For examples of our last interaction, the interpolation interaction, recall
that a \emph{kernel}
in a graph $G=(V,E)$ is a subset $K\subseteq V$
that satisfies two properties:
(1) Independence: for every $u,v\in V$, $(u,v)\notin E$;
(2) Absorbance: for every $u\in V\setminus K$,
there exists $v\in K$ such that $(u,v)\in E$.
In other words, a kernel is a maximal independent set.

\begin{example}
\label{ex:interpolationSingle}
{\textsc{Single interpolation track.}}
Consider the interactions mix console that again
consists of one single track that is turned \fbox{\colorbox{OliveGreen}{\textsf{\tiny \color{white}On}}}
and that carries the interpolation track $\Phi^{(k)} = \Phi$ described
in subsection \ref{subsec:interpolationTrack}, that is, $\Phi$ is defined through
equations \eqref{eq:interpolTrackPlus} and \eqref{eq:interpolTrackMinus}.
\begin{itemize}
\small
\item
On startup, the Ionian scale is selected and the MCMC is not running.
Freeze the selected scale by pressing the \fbox{\textsf{\tiny F}} key.
Next, press the \fbox{\textsf{\tiny +}} and \fbox{\textsf{\tiny -}} buttons to set
$J^+,J^-= \{1,2,3,\ldots\}$ and activate the external magnetic field by pressing
the \fbox{\textsf{\tiny M}} key once (a notification that the external magnetic field
is activated is shown in the GUI). Start the MCMC machine and observe how the
frozen Ionian scale transmits its energy to its neighbors for the configurations
converge to the Ionian scale together with the set of scales that interpolate the
Ionian scale in any of the $\mathfrak I^{k\pm}$-step interpolation networks.
\item
Here is another way to obtain the configuration of the previous item (without the external magnetic field).
Again after startup, Freeze the selected scale (Ionian).
Start the MCMC machine and press the \fbox{\colorbox{subspace}{\textsf{\tiny \color{white} S7}}}
button, that is, run the Subspace algorithm on the 7th column of the bell of scales
that consists of an i.i.d. \textsf{Bernoulli}$(1/2)$ random field except of the selected scale
(e.g. Ionian) that is frozen in the \textsf{On} state. As a result we get
$\mathcal X = \mathcal C_1\cup\cdots \cup\mathcal C_6
\cup\{\textnormal{Ionian}\}\cup
\mathcal C_8\cdots \cup \mathcal C_{12}$.
Set $J^+= \{1\}$ and then, using the equalizer,
let the interpolation track have effect only on scales of length $\leq 6$.
The generic configurations will have scales of length less that 6
in the \textsf{On} state whenever they $j^+$-step interpolate the Ionian scale.
Press the buttons \fbox{\colorbox{subspace}{\textsf{\tiny \color{white}Sj}}} and
\fbox{\colorbox{freeze}{\textsf{\tiny \color{white} Fj}}} for $j=1,\ldots , 6$ to set
$\mathcal X = N_{\mathfrak I^{6+}}(\textnormal{Ionian}) \cup \cdots \cup 
N_{\mathfrak I^{1+}}(\textnormal{Ionian}) \cup\{\textnormal{Ionian}\}\cup
\mathcal C_8\cdots \cup \mathcal C_{12}$
and $\mathcal F = N_{\mathfrak I^{6+}}(\textnormal{Ionian}) \cup \cdots \cup 
N_{\mathfrak I^{1+}}(\textnormal{Ionian}) \cup\{\textnormal{Ionian}\}$.
Similarly, now let $J^+= \varnothing$ and $J^- = \{ 1 \}$ and then
let the interpolation track have effect only on scales of length $> 6$.
\item
Set $J^+,J^-= \{1\}$ (hence the underlying network is the
(undirected) interpolation network $\mathfrak I$. If all the volumes are all the way up, then
a scale will tend to be in the \textsf{On} state if it interpolates
with another scale that is already in the \textsf{On} state, otherwise
it will tend to be in the \textsf{Off} state. Thus, under this mix, starting from a
generic random configuration, the MCMC machine will surely converge to
the configuration with all the scales in the \textsf{On} state.

But now let us bring the main volume $-\beta \gg 0$ all the way down.
In this case we do not get the complement (i.e. the configuration with all the scales in the \textsf {Off}
state). Instead, we get configurations of scales with the property that the set of scales in
the \textsf{On} state form a maximal independent set in $\mathfrak I$. In other words, we obtain
kernels in the hypercube of dimension 11. In Figure \ref{fig:HistogramKernelHyper}
we show a histogram for the cardinalities of a random sample of kernels obtained with \emph{Scaletor}.

\begin{figure}[htbp] 
   \centering
   \includegraphics[width=4.7in]{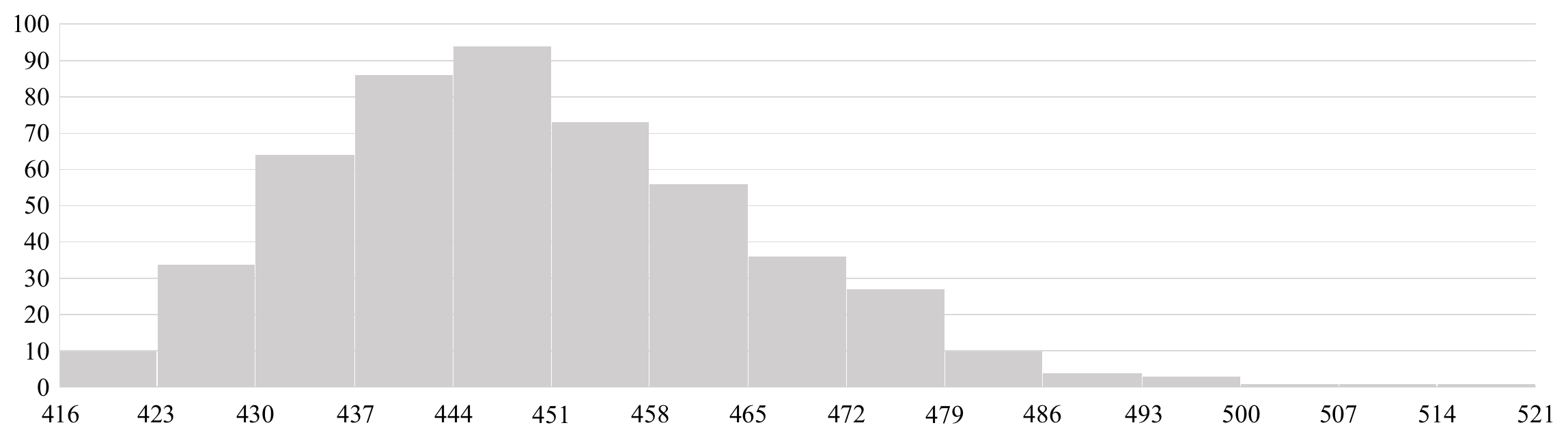} 
   \caption{A histogram of the cardinalities of a random sample of
   kernels in the hypercube of dimension 11,
   produced by the interpolation track mixed as in Example
   \ref{ex:interpolationSingle}
   (the size of the sample is $500$).}
   \label{fig:HistogramKernelHyper}
\end{figure}
\end{itemize}
\end{example}

\begin{example}
\label{ex:firstSymbolRule}
{\textsc{First symbol rule.}}
With \emph{Scaletor} we can verify results from
\cite{GomezNasser21, Gomez21,Gomez23} for example.
Indeed, recall subsection \ref{subsec:fsr}. The default content in the
file \texttt{sequence.txt} is an initial segment of the well known Thue-Morse sequence.
So if the
\fbox{\colorbox{fsr}{\textsf{\tiny \color{white}First Symbol Rule}}}
button is pressed, then we obtain all the 18 scales
reported in \cite{GomezNasser21}. Freeze them by pressing the
\fbox{\colorbox{freeze}{\textsf{\tiny \color{white}Freeze}}} button.
Activate the magnetic field by pressing the key  \fbox{\textsf{\tiny M}} in the keyboard.
Turn \fbox{\colorbox{OliveGreen}{\textsf{\tiny \color{white}On}}} the modes track
with $j=1$ and the $\geq$ case, and set all the parameters $\gg 0$
to obtain the modes orbitals of the
Thue-Morse scales (and there are indeed 49 scales in total). The same can be done for other sequences
like the ones that result from the Fibonacci and Feigenbaum substitutions
and thus we can readily verify all the claims in the aforementioned
references. (Try finding scales arising from the digits of
$\pi\triangleq 3.14159\ldots$.)
\end{example}

\begin{example}
\label{ex:multitrack}
{\textsc{Multitrack mix.}}
Turn \fbox{\colorbox{OliveGreen}{\textsf{\tiny \color{white}On}}} the interval tracks
of thirds, both minor and major thirds (that is, when $m = 3$ and $m=4$), with
both $j=1$ and the $\geq$ case for both tracks, with all the sliders $\gg 0$ on both track too,
and start the MCMC with $\beta \gg 0$. As a result we obtain the 282 integer compositions of 12 
with at least one 3 and at least one 4 because by default both tracks are multiplicative.
Now make both tracks additive by clicking each in its own area where none of its buttons
sliders are present (the background of the track will turn dark red). We will see the
random configurations of scales that are accepted whenever the following two conditions are satisfied:
\begin{enumerate}
\small
\item
\label{ex:case1}
if a site $s\in\mathcal C_n$ has no $3$ or no $4$, then it is in the \textsf{Off} state;
\item
\label{ex:case2}
if a site $s\in\mathcal C_n$ has both at least one $3$ and at least one $4$, then it is in the \textsf{On} state.
\end{enumerate}
We already had access to the 282 scales that satisfy condition \ref{ex:case2}, but
we could get them back again by pressing the 
\fbox{\colorbox{subspace}{\tiny \color{white}\textsf{Subspace}}} button
(several times if necessary).
Press the
\fbox{\colorbox{subspace}{\tiny \color{white}\textsf{Release Subspace}}}
button in case you pressed the
\fbox{\colorbox{subspace}{\tiny \color{white}\textsf{Subspace}}}
button to come back to random configurations of scales that satisfy both conditions
\ref{ex:case1} and \ref{ex:case2}.
To get the 1520 scales that satisfy \ref{ex:case1}, press the
\fbox{\colorbox{freeze}{\textsf{\tiny \color{white}Freeze}}} button
(several times if necessary). Press the 
\fbox{\colorbox{freeze}{\textsf{\tiny \color{white}Release Freeze}}} button
to go back again to random configurations of scales that satisfy both  conditions
\ref{ex:case1} and \ref{ex:case2}.
Turn the
modes track with $j=1$ and the $\leq$ case, it imposes on the following condition
on random configurations of scales:
\begin{enumerate}
\small
\setcounter{enumi}{2}
\item
\label{ex:case3}
a site $s\in\mathcal C_n$ can be in the \textsf{On} state only when all the
other members of its modes class are in the \textsf{Off} state.
\end{enumerate}
By default the modes track is multiplicative (green background)
and at the current moment the interval track controlling the occurrences
of thirds should be both additive (dark red backgrounds). Hence, we get random configurations
of scales that satisfy conditions \ref{ex:case1}, \ref{ex:case2} and \ref{ex:case3}.
Press back and forth the
\fbox{\colorbox{freeze}{\textsf{\tiny \color{white}Freeze}}} and
\fbox{\colorbox{freeze}{\textsf{\tiny \color{white}Release Freeze}}} 
buttons (remember to wait a few seconds after pressing the former button)
and get, in the terminology of \cite{Gomez21}, transversals to the set
of scales that satisfy condition \ref{ex:case1}, and observe that
the transversal dimension of this set is 259. To obtain transversals to
the set that satisfies condition \ref{ex:case2}, make the modes track
additive and the interval tracks mutiplicative, and press the
\fbox{\colorbox{reset}{\textsf{\tiny \color{white}Reset}}} button
several times to produce transversals to this set of transversal dimension 52.
\end{example}

\section{Conclusions and related work}
\label{sec:conclusion}

The paradigm of using statistical mechanics as a model to
describe emerging order in disordered systems has
been successfully exported to fields like neurosciences, economics,
information theory and machine learning, to mention just a few
(in social networks for example, the Ising model has been adapted
to study phenomena like group polarization, echo chamber and cocoon
effects \cite{DaiZhuWang22, NiGuidiMichienziZhu23, ZhuNiWangLi21}).
Music is not the exception, e.g. Euler's Tonnetz have been used
to exhibit how ordered patterns emerge in music \cite{Berezovsky19}
(in this regard, generalizations of Tonnetz are very natural grounds for further exploration,
consider for example \cite{Mohanty22, Yust20}; see also \cite{Aucouturier08}
where self-organization is put in the context of musical tuning systems).
In particular, the Ising model has been adapted in musical contexts like sonification
\cite{Campo2007, ClementeCrippaJansenTuysuz22} and
restricted Boltzmann machines have been used to study musical sequences
\cite{LattnerMaartenAgresCancino15}.

\emph{Scaletor} is not only a complete catalog of musical scales in any given
tonality, its scope encompasses educational purpose for it constitutes
an interactive tool that serves to understand the basic principles underlying
these type of models of thermodynamic formalism that are certainly
demanding\footnote{See the foreword to the first edition of \cite{Ruelle04}.},
at least when someone is first exposed to the subject. We find this kind of
approach in other works, e.g. software applications have been developed in
quantum computer music with educational purposes in particular (see \cite{Weaver22}
as part of the whole volume \cite{Miranda2007}).
Here we have effectively adapted the scenario of statistical
mechanics to the context of classification of musical scales, which
are central objects in music theory and still constitute
a subject of current research (see e.g. \cite{Nuno21, Elliot22},
also \cite{HunterHipper03, Hook07, ClampittNoll18, Zheng23}).
The outputs can provide guidelines that musicians can use for
compositions, methods of study, and so forth. Furthermore,
we have seen that the network structures implemented in \emph{Scaletor}
allow us to find kernels in (directed) graphs such as the hypercube
(see \cite{Galvin11, Galvin12, Jenssen20, Jenssen22}, also
\cite{GalvinTetali06, Restrepo14} where independent sets are considered together
with Glauber dynamics; also, for more on networks in music see e.g.
\cite{GrantKnightsPadillaTidhar22, PopoffYust22}).

\emph{Scaletor} comes with a handful set of interaction tracks,
many of which are related to the combinatorics of integer compositions
(for more on combinatorial problems in the theory of music see e.g. see \cite{Read97, Peck21}),
but we have also included balance which is a geometric property, and
certainly there is room for further exploration in this regard (see e.g.
\cite{Tymoczko06, MathiasAlmada21}). We have seen just a few basic examples but
the possibilities are vast and more elaborate mixes require careful interpretations.
Furthermore, although here we have
focused solely on musical scales, the underlying structure applies 
just as well to rhythm structures. \emph{Scaletor} has been designed to
escalate, adding new tracks is a relatively simple process and the code
has been conceptually written in a way that the condition
$n=12$ is not a restriction, so a project that incorporates both scales and
rhythms together with more and multiple musical interactions can be developed,
even with with MIDI capabilities.

\section{Additional material}
\label{sec:additional}

\begin{itemize}
\item
\emph{Scaletor} is available for download from its repository at GitHub: \medskip

\centerline{\texttt{https://github.com/gomiza/scaletor}}

\medskip
The README file contains information about installation, launch instructions and additional material like video tutorials.
\end{itemize}

\section*{Declarations}

{\bf Funding.} This work was supported by DGAPA-PAPIIT grants IN107718, IN110221,
IN112725.

\bigskip

\noindent
{\bf Conflict of interest.} No potential conflict of interest was reported by the author.

\bigskip

\noindent
{\bf Acknowledgments.} Our initial motivation arouse while the author
was working on other projects from thermodynamic formalism that were inspiring:
we are thankful to Brian Marcus, Siamak Taati,
Sebastián Barbieri and Tom Meyerovitch for multiple discussions.
Special thanks to the musical ensamble \emph{Los Simb\'olicos Din\'amicos}
that has composed, arranged and performed musical pieces based
on this and other works that explore the relation of mathematics and music,
their members have included Dante Baz\'ua, 
Aldo Max, Gustavo Rivera and Jos\'e Luis Vaca.


\bibliography{RFS}


\end{document}